\documentclass[11pt,draft]{article}
\usepackage{fullpage}
\usepackage{amsmath, amssymb, amsthm, amsfonts, latexsym, color, paralist, cancel,ifdraft,mathtools}
\usepackage[final]{hyperref}
\usepackage[final]{graphicx}
\usepackage{scalerel,stackengine}
\stackMath
\newcommand\reallywidehat[1]{%
\savestack{\tmpbox}{\stretchto{%
  \scaleto{%
    \scalerel*[\widthof{\ensuremath{#1}}]{\kern.1pt\mathchar"0362\kern.1pt}%
    {\rule{0ex}{\textheight}}%
  }{\textheight}%
}{2.4ex}}%
\stackon[-6.9pt]{#1}{\tmpbox}%
}
\usepackage{authblk}
\usepackage[x11names]{xcolor}
\usepackage[modulo]{lineno}
\usepackage{cleveref}
\usepackage{float}
\usepackage{tikz}
\usetikzlibrary{positioning, arrows,calc, fit, shapes, chains, matrix}
\usetikzlibrary{decorations.pathreplacing}
\usepackage{tikzsymbols}
\usepackage{ifthen}
\usepackage{subcaption}
\usepackage{adjustbox}

\ifdraft{}{}
\usepackage[refpage]{nomencl}
\makenomenclature

\hypersetup{
  colorlinks   = true, %
  urlcolor     = blue, %
  linkcolor    = {blue!55!black}, %
  citecolor   = red %
}

  \theoremstyle{plain}
  \newtheorem{theorem}{Theorem}[section]
  \newtheorem{lemma}[theorem]{Lemma}
  \newtheorem{corollary}[theorem]{Corollary}
  
  \newtheorem{claim}[theorem]{Claim}
  \theoremstyle{definition}
  \newtheorem{definition}[theorem]{Definition}
    \newtheorem{remark}[theorem]{Remark}
  
  \newtheorem{example}{Example}
  \newtheorem{conjecture}{Conjecture}

\newcommand{\comment}[1]{}

\newcommand{\problem}[1]{\ensuremath{\mathsf{#1}\xspace}} %

\newcommand{\defeq}{\stackrel{\mathrm{def}}{=}}

\newcommand{\norm}[1]{\left\| #1 \right\|}
\newcommand{\bra}[1]{\langle #1 \vert}

\newcommand{\ket}[1]{\vert #1 \rangle}

\newcommand{\ketbra}[2]{\vert #1 \rangle \langle #2 \vert}
\newcommand{\kb}[1]{\vert #1 \rangle \langle #1 \vert}
\newcommand{\braket}[2]{\langle #1 \vert #2 \rangle}

\newcommand{\eps}{\varepsilon}
\newcommand{\real}{\mathbb{R}}
\newcommand{\sizepath}{\ensuremath{k^{\frac{2kd}{\eps} \log_{\left(1 + \frac{\eps}{4}\right)}|\Sigma|}}}

\newcommand{\lightcone}[1]{#1^{\triangle}}
\newcommand{\outsidelightcone}[1]{#1'}
\newcommand{\globalproj}[1]{\tilde{#1}}
\newcommand{\localproj}[1]{#1}

\newcommand\defclass[5]{
\begin{definition}[#1]\label{#2}
#3 \\
\;\;\textbf{Completeness.} #4 \\
\quad\quad\textbf{Soundness.} #5
\end{definition} 
}

\newcommand\defproblem[5]{
\begin{definition}[#1]\label{#2}
#3 \\
\hspace{1cm} \textbf{Yes.} #4 \\
\noindent\;\;\textbf{No.} #5
\end{definition} 
}

  \newcommand{\class}[1]{\textup{#1}}
  \newcommand{\cP}{\class{P}} 
  \newcommand{\PCP}{\class{PCP}} 
  \newcommand{\NP}{\class{NP}} 
  \newcommand{\MA}{\class{MA}}   
    
  \newcommand{\RP}{\class{RP}}

    \newcommand{\complex}{\mathbb{C}}

 \newcommand{\ayes}{A_{\rm yes}} 
  \newcommand{\ano}{A_{\rm no}} 
\usepackage{authblk}
  
\begin{document}

\pagenumbering{gobble}

\title{Stoquastic PCP vs. Randomness}
\author[1]{Dorit Aharonov}
\author[2]{Alex B. Grilo}
\affil[1]{Hebrew University of Jerusalem, Jerusalem, Israel}
\affil[2]{CWI and QuSoft, Amsterdam, The Netherlands}
\date{}
\maketitle

\begin{abstract}
The derandomization of MA, the probabilistic version of NP, is a long standing open question. In this work, we connect this problem to a variant of another major problem: the quantum PCP conjecture.  
Our connection goes through the surprising quantum characterization of MA by Bravyi and Terhal. They proved the MA-completeness of the problem of deciding whether the groundenergy of a uniform stoquastic local Hamiltonian is zero or inverse polynomial. We show that the gapped version of this problem, i.e. deciding if a given uniform stoquastic local Hamiltonian is frustration-free or has energy at least some constant $\epsilon$, is in NP. Thus, if there exists a gap-amplification procedure for uniform stoquastic Local Hamiltonians (in analogy to the gap amplification procedure for constraint satisfaction problems in the original PCP theorem), then MA = NP (and vice versa). 
Furthermore, if this gap amplification
  procedure exhibits some additional 
(natural) properties, then P = RP.
We feel this work opens up a rich set of new directions to explore, which might lead to progress on both quantum PCP and derandomization.  

We also provide two small side results of potential interest. First, we are able to generalize our result by showing that deciding if a uniform stoquastic Local Hamiltonian has negligible or constant frustration can be also solved in NP. Additionally, our work reveals a new MA-complete problem which we call SetCSP, stated in terms of classical constraints on strings of bits, which we define in the appendix. As far as we know this is the first (arguably) natural MA-complete problem stated in non-quantum CSP language.
\end{abstract}
\newpage
\pagenumbering{arabic}
\section{Introduction}
It is a long standing open question, whether the randomized version of NP,
called MA (for Merlin-Arthur) can be derandomized, namely, whether MA = NP.
MA is defined like NP, except the verifier is probabilistic: If $x$ is a yes-instance, there exists a proof $y$ of polynomial size, such that the verifier always accepts $(x,y)$. If $x$ is a no-instance, the verifier rejects $(x,y)$ for any polynomial $y$, with high probability. 

The derandomization of MA is implied by widely believed conjectures such as NEXP does not have polynomial size circuits~\cite{ImpagliazzoKW02},
as well as by the stronger P =
BPP conjecture~\cite{Goldreich2011}, which itself is implied by the existence of one-way
functions~\cite{HastadILL99} and also by commonly conjectured  circuit lower
bounds~\cite{BabaiFNW93,NisanW94,ImpagliazzoW97,SudanTV01,KabanetsI2004}. The
(somewhat counter-intuitive at first) connection between lower-bounds on
computation, and derandomization (which can be seen as an upper-bound result)
coined the intriguing name ``Hardness versus Randomness''~\cite{NisanW94}. Our work follows this path, and provides a result 
of a similar flavor: we connect the derandomization of 
MA (as well as, more weakly, to that of RP) with a different hardness problem in computational complexity -- that of quantum 
PCP -- hence the title of this paper. 

Our starting point is a seminal and beautiful paper of Bravyi and Terhal
\cite{BravyiT09}, where they introduced 
the first natural MA complete problem, 
which surprisingly is defined in the quantum setting\footnote{ \label{fn:ma-complete}
For PromiseMA, it is folklore that one can define complete problems by extending NP-complete problems (see, e.g. \cite{WilliansShor19}):
we define an exponential family of 3SAT formulas (given as input succinctly) and 
we have to decide if there is an
assignment that satisfies all of the formulas, or for every assignment, a random
formula in the family will not be satisfied with good probability.}.  
To explain this quantum characterization of MA, and how we use it to make the connection to quantum PCP, let us make a detour to quantum Hamiltonian complexity.

\subsection{Hamiltonian complexity and stoquastic Hamiltonians} 
The power of QMA~\cite{KitaevSV02} the quantum version of NP, had been a major area of study in the past decade, leading to enormous 
progress in our understanding of the complexity of quantum states and the reductions between Hamiltonians~(see \cite{Osborne12,GharibianHLS15}). In 
QMA, a polynomial time {\it quantum} verification algorithm receives a {\em quantum}
proof $\ket{\psi}$ for some classical input $x$. and applies a quantum polynomial algorithm on both $x$ and $\ket{\psi}$. If $x$ is a yes-instance, there is a $\ket{\psi}$ which causes 
the verifier to accept with high probability; 
otherwise, no matter what the quantum proof $\ket{\psi}$ is, $x$ should be rejected with high probability.  In addition to being a natural generalization of classical proof
systems, the relevance of QMA was evidenced by Kitaev, who showed that
estimating the lowest eigenvalue of a local Hamiltonian, a central problem in many body quantum physics, is
complete for QMA~\cite{KitaevSV02}. Kitaev's theorem is the quantum analog of the Cook-Levin theorem \cite{Cook71,Levin73}, 
and it makes a strong connection between 
a major question in condensed matter physics (namely, groundstates of local Hamiltonians), and a major problem in Theoretical Computer Science, (namely, optimal solutions for constraint satisfaction problems).
In fact, the connection is even deeper since what is shown is that the latter is simply a special case of the former.  

More concretely\footnote{See \Cref{sec:shamiltonian} for a detailed
definition.}, physicists associate with an $n$ qubit system a self-adjoint operator called a Hamiltonian, which corresponds to the {\it energy} of the system, and can be
usually decomposed as a sum of terms corresponding to interactions between a small number of particles. Of major interest is the lowest eigenvalue of this operator, and its corresponding eigenstates - called {\it groundstates}. Looking at this problem through the Theoretical Computer Science lens, 
Kitaev defined the $k$-Local Hamiltonian problem \cite{KitaevSV02, AharonovN02}, whose input is a Hamiltonian $H$ on an
$n$ particle system given as a sum of $m$ local terms, each of
them acting non-trivially on at most $k$ out of the $n$ particles (the term local only refers to the fact that $k$ is assumed to be small, there are no geometrical restrictions on 
the interaction). We are also given as input two parameters, $\alpha$ and $\beta$. 
We then ask 
if the smallest eigenvalue of $H$ is smaller  than $\alpha$, or all states have energy larger than $\beta$. The hardness of the Local Hamiltonian problem depends on
the input {\it promise gap} defined as $\beta - \alpha$; Kitaev showed that the
problem is QMA-complete for some inverse polynomial promise gap~\cite{KitaevSV02, AharonovN02}. 

Bravyi, DiVincenzo, Oliveira and Terhal~\cite{BravyiDOT08} asked how
the problem behaves when the terms are restricted such that their off-diagonal elements are all non-positive, a property that they named ``stoquastic'' (as a combination of the words stochastic and quantum). This
property implies a lot of structure on groundstates (See
\Cref{lem:gspace-structure}), 
and in physics it is associated with the lack of the ``sign problem'', in which case one can associate 
with the Hamiltonian a classical Markov Chain Monte Carlo experiment and study it (See \cite{BravyiDOT08}); such systems 
are considered far easier than general Hamiltonians. \cite{BravyiBT06} showed that the
stoquastic Local Hamiltonian problem is StoqMA-complete, where StoqMA is a complexity class that sits between MA and QMA. 

Importantly for this paper, Bravyi and Terhal~ \cite{BravyiT09} then showed that the stoquastic Local
Hamiltonian problem is MA-complete if we pick $\alpha = 0$ and $\beta \geq
\frac{1}{poly(n)}$, or in other words, if we want to decide whether the Hamiltonian is {\it frustration-free}\footnote{A Hamiltonian is frustration-free if there is a state in the groundspace of 
all local terms.} or the lowest eigenvalue is at least inverse polynomial. This was, to the best of  our knowledge, the first MA-complete problem which is not an extension of NP-complete problems into the randomized
setting\footnote{As commented on in \Cref{fn:ma-complete}}  (see also 
\cite{Bravyi14}). A simple  observation important for the current work is that in the MA-complete problem of \cite{BravyiT09}, the groundspaces of the local terms are all spanned by {\it subset-states}, i.e., states which are the {\it uniform} superposition of a subset of strings. We call such Hamiltonians {\it uniform} stoquastic Hamiltonians.

This paper is concerned with the {\it gapped}
version of the uniform stoquastic Local Hamiltonian problem. 
Gapped versions of 
NP-hard problems have played a crucial role in the 
topic of
probabilistically checkable proofs (PCPs) which had
revolutionized classical Theoretical Computer Science over the past three decades. Before we define the gapped version of the uniform stoquastic Hamiltonian problem, let us introduce the by-now-standard notion of PCPs in more detail.

\subsection{The PCP theorem}
The ``mother'' of all NP-complete problems is 3SAT. An instance to this problem is a
Boolean formula $\phi$ in the form $\phi(x) = C_1 \wedge C_2 \wedge ... \wedge
C_m$, where $C_i = (y_{i, 1} \vee y_{i,2} \vee y_{i,3})$ is a clause and
$y_{i,j} \in
\{x_1,...,x_n, \overline{x_1}, ... \overline{x_n}\}$.  We ask if there exists an
assignment $x \in \{0,1\}^n$ such that $\phi$ is satisfied.

The problem MAX3SAT$_\delta$, parametrized by some function $\delta(n)$, is a
generalization of 3SAT. In this problem we have to distinguish between the cases
where $\phi$ is satisfiable or for every assignment of the input variables, at
least a $\delta(n)$ fraction of the clauses are not satisfied.

By picking $\delta(n) = \frac{1}{m}$, MAX3SAT$_\delta$ becomes equivalent to 3SAT,
and therefore it is NP-complete. In PCPs we are 
interested in versions of the problem with significantly larger $\delta$. The celebrated PCP theorem~\cite{AroraLMSS98,AroraS98} states the remarkable result that there exists some
constant $\eps$ independent of $n$, such that the problem MAX3SAT$_\eps$ is
NP-complete. This problem with constant $\eps$ is called the {\it gapped} version of the problem, and the PCP theorem proves that the gapped version of this problem is as hard as the original one. 

In her celebrated alternative proof to the PCP theorem, Dinur \cite{Dinur07}
used an explicit {\it gap amplification} procedure, or {\it reduction}. The reduction takes an $n$-bit instance $\phi$ of 3SAT (or equivalently  of
MAX3SAT$_{\frac{1}{m}}$) to an instance $\phi'$ of
MAX3SAT$_{\eps}$, acting on not many more than $n$ bits, for some constant $\eps$. It is required that $\phi$ is
satisfiable iff $\phi'$ is satisfiable. Importantly, when $\phi$ is not
satisfiable, every assignment to the variables of $\phi'$ leaves at least
$\eps$ fraction of the clauses of $\phi'$ unsatisfied. This shows that even when there is a constant $\eps$ promise gap between the yes- and no-cases, the problem remains NP-hard. 
The PCP theorem is one of the crown jewels of Computational
Complexity Theory, with far-reaching applications such as hardness of approximation, verifiable delegated computation, program obfuscation, and cryptocurrencies (e.g., \cite{Hastad2001, GoldwasserKR15, BonehISW17,Ben-SassonCG0MTV14}).

Does a quantum version of the PCP theorem ~\cite{AharonovN02,AharonovALV09,AharonovAV13} hold? The gap amplification version of the 
quantum PCP conjecture can be stated as follows: the Local Hamiltonian problem remains QMA-complete even when the promise gap is a constant $\eps$. This PCP conjecture has implications also to our understanding of 
multipartite entanglement, specifically, whether multipartite entanglement can persist at ``room-temperature'' (see \cite{AharonovAV13}). Despite much work on 
this direction \cite{AharonovALV09, Arad11,Hastings13, BrandaoH13a, FreedmanH14,AharonovE15,EldarH17,NirkheVY18}, progress had so far been limited, and it remains a major open problem \footnote{
We mention that a the {\it game} version of quantum PCP, whose connection to the above quantum PCP is not well understood, was recently resolved \cite{NatarajanV18}.}. 

What about PCPs for MA? To the best of our knowledge, the only work on PCPs for randomized classes is Drucker's~\cite{Drucker11}, which proves a PCP theorem for AM, another randomized version of NP. 
The following is a complete problem: an instance is a (succinctly given) {\it collection} of Boolean
formulas $\{\phi_r\}$, and we want to decide between the yes-case: every formula in this family
is satisfiable, or the no-case: with
high probability a uniformly random formula in this family is not satisfiable. Drucker proved that if one replaces "not satisfiable" in the description of no instances, by "$\epsilon$-far from satisfiable" (for some fixed constant $\epsilon$) - namely every assignment violates at least $\epsilon$ fraction of the constraints - the problem remains AM hard. 
The proof relies on the PCP theorem for NP, 
and does not seem to provide insight as to how to carry it over 
to MA. 

\subsection{Our Contribution}
This work creates a surprising link between the long standing problem of derandomizing MA, to the seemingly 
unrelated question of quantum PCP for stoquastic Hamiltonians. By the work of \cite{BravyiT09} the latter question can be viewed as a variant of PCP for MA. Before stating the results, let us 
provide some definitions.

\paragraph{Definition (Informal):} The ($\epsilon,k,d$)- Gapped, Uniform, Stoquastic, Frustration-Free, Local Hamiltonian problem.)
 An input is a set of $m$ positive semi-definite uniform stoquastic terms $H_1, \ldots, H_{m}$, where each $H_i$ acts on $k$ out of the $n$ qudit system and $\norm{H_i}
 \leq 1$; moreover every qubit participates in at most $d$ terms. 
For $H = \frac{1}{m} \sum_{j = 1}^{m} H_j$, decide which of the following holds, given the promise that one is true: \\
\hspace{1cm} \textbf{Yes.} There exists a quantum state
       $\ket{\psi}$ such that
      $\bra{\psi} H \ket{\psi}
        =0$. (the groundenergy of $H$ is $0$) \\
\noindent\;\;\textbf{No.} For all quantum states $\ket{\psi}$
      it holds that
      $\bra{\psi} H \ket{\psi}
        \geq \eps$ . (all eigenvalues are at least $\epsilon$)

{~}

Thus, we are given a bounded degree uniform stoquastic $k$-local Hamiltonian which is
either frustration-free or constantly frustrated. 
Our main result is that distinguishing these two 
is in NP.

\newcommand{\bodymainthm}{For any constants $\eps>0,k,d$, the problem of deciding whether a uniform $d$-bounded degree stoquastic $k$-local Hamiltonian $H$ is frustration-free or $\eps$-frustrated, is in \NP{}.}
\begin{theorem} [{\bf Main: uniform stoquastic frustration-free Local Hamiltonian
  is in NP}]
\label{thm:main}
\bodymainthm
\end{theorem}

We note that the same problem, except with inverse polynomial gap, is MA-complete. 

The restriction in \Cref{thm:main} 
that the yes-instances have perfect completeness (frustration free) seems too strong. Indeed, further work enables us to strengthen \Cref{thm:main} and relax the requirement, such that yes-instances are just negligibly frustrated: 

\newcommand{\bodythmfrustrated}{The problem of
deciding whether a uniform stoquastic Hamiltonian $H$ on $n$ qubits has negligible frustration\footnote{A function $f$ is called {\em negligible} if  $f = o\left(\frac{1}{n^c}\right)$ for every constant $c$.} or is at least $\eps$-frustrated is in \NP{}, for any 
constant $\epsilon$.}
\begin{theorem} [{\bf uniform stoquastic Local Hamiltonian with Imperfect completeness is in NP}] 
\label{thm:exponentially-small}
\bodythmfrustrated
\end{theorem}

Hence, first of all, this provides a new tighter upper-bound 
on the hardness of ground energy and groundstate of stoquastic Hamiltonians, in case the promise gap is constant. 
This is of interest first of all in the context of quantum Hamiltonian complexity \cite{Osborne12,GharibianHLS15}.
We next discuss the implications of these results from the complexity-theoretical
point of view. To this end we propose the following conjecture: 
\begin{conjecture}[Stoquastic PCP conjecture] (Informal) 
\label{conj:pcp}
There exist constants $\eps > 0$, $k',d' > 0$ and an efficient gap amplification procedure that reduces the problem deciding if a uniform stoquastic $d$-degree $k$ Local-Hamiltonian is frustration-free or at least inverse polynomially frustrated, to the problem of deciding if a uniform stoquastic $d'$-degree $k'$ Local-Hamiltonian is frustration-free or at least $\eps$ frustrated.

\end{conjecture}

An immediate consequence of Theorem \ref{thm:main} is that Conjecture \ref{conj:pcp}  implies MA = NP.

\begin{corollary}[uniform stoquastic PCP conjecture implies derandomization of \MA{}]\label{cor:pcp-derandomization}
If the stoquastic \PCP{} conjecture is true, then \MA{} = \NP{}. 
\end{corollary}

From an optimistic perspective, our result thus opens a way towards proving that MA = NP (as well as circuit lower bounds 
by \cite{ImpagliazzoKW02}) via quantum arguments, in particular by proving specific types of quantum PCPs. This path could of course be very hard, but under the belief that MA = NP, such a gap amplification procedure is in fact {\it known} to exist. 
Taking the opposite point of view, and assuming the less commonly believed assumption that 
MA is strictly larger than NP, our work proves that 
no PCP exists for stoquastic local Hamiltonians, or,  loosely speaking, there is no PCP for MA.

It is natural to ask whether these results have implications 
to the derandomization of RP and BPP, and in particular, 
would stoquastic quantum PCPs imply anything in this scaled down context. 
We provide some (weaker, namely with more assumptions) results in this direction, stated in Appendix \ref{sec:coRP}. 
We describe these briefly here. 
First, one can slightly modify the uniform stoquastic Local Hamiltonian problem
by requiring that in the yes-case, the groundstate 
has the all $0$ string in its support (see \Cref{def:pinned-local-hamiltonian}); then no witness is needed.
We call this the {\it pinned} version of the problem. 
The same proof of Theorem \ref{thm:main} would imply that the gapped version of this pinned uniform stoquastic problem is in P. 
A somewhat strengthened version of Conjecture \ref{conj:pcp}, with the additional natural requirement that the gap amplification reduction is also {\it witness preserving}~\footnote{A reduction is witness preserving if
there is a polynomial time algorithm that maps a witness of the original problem into a witness of the target problem.}
would then imply that P  = co-RP, and since P is closed under complement, P = RP. 
One might expect that Theorem 
\ref{thm:exponentially-small} would imply 
similar implications for BPP; 
however we do not know how to overcome 
a technical obstacle in this argument 
and thus clarifying a reasonable set of assumptions  that would imply BPP=P remains open. 
For details see \Cref{sec:coRP}.

Finally, as a small side result, we present in \Cref{sec:commuting} an alternative proof for the result of~\cite{GharibianLSW15} that shows that  the {\it commuting} version of stoquastic Hamiltonian problem is in NP (for any promise gap.) 
\newcommand{\bodycommutingthm}{The problem of deciding if a commuting stoquastic
Hamiltonian $H$ is frustration-free is in \NP{}.}

\begin{theorem}[commmuting stoquastic Local Hamiltonian problem is in NP]
\label{thm:commuting} \bodycommutingthm \end{theorem}

\subsection{Proof overview and main ideas}

We prove \Cref{thm:main} by derandomizing the verification algorithm used in~\cite{BravyiT09} in their proof of the containment in MA of the inverse-polynomial
version of the stoquastic Hamiltonian problem. 
The derandomization becomes possible when the 
gap is constant, namely, when we know 
that the Hamiltonian is either frustration free or there is a large amount of frustration. 

We briefly explain now the main ideas behind
the randomized verification procedure of \cite{BravyiT09}, using a random walk; then we overview our approach to derandomize it.
They start by defining an (exponential-size) graph whose vertices are all possible $n$-bit
strings. The edges are defined based on the
stoquastic Hamiltonian: two strings $x,y$ are adjacent in the graph, iff they
are connected by some $H_i$, one of the local terms of the Hamiltonian, namely,
if $x$ and $y$ appear together in some groundstate of $H_i$.
The paper considers the following random-walk on the graph: starting from a given $n$-bit string, pick one of the terms uniformly at random, and go to any of the (constantly many) strings 
connected to the current string by that term, uniformly at random. 
This is called a {\it step}.\footnote{In the non-uniform case, there are weights involved and the random-walk becomes more complicated, but here we focus only on the uniform case. }
Bravyi and Terhal also define the notion of a {\it bad string},
which is a string that does not appear in the support of any of the groundstates
of some local term. All other strings are {\it good}. 

Bravyi and Terhal then showed that if the stoquastic Hamiltonian is
frustration-free, then the connected component of any string in the support of
some groundstate of the Hamiltonian, does not contain bad strings. In particular, 
any walk on the above defined graph, starting from some string in a groundstate, does not reach bad strings.
On the other hand, if the Hamiltonian is at least $\frac{1}{p(n)}$ frustrated,
for some polynomial $p$, then there exists some polynomial $q$ such that a
$q(n)$-step random walk starting from any initial string reaches a bad string
with high probability.
The MA verification algorithm then proceeds by the Prover sending some $x$, which is supposed to lie in the support of some groundstate
of the Hamiltonian and the verifier performs a $q(n)$-step random
walk starting from $x$, as above. The algorithm rejects if a bad string is encountered in the random-walk. 

Our main technical result is showing that if the Hamiltonian is $\eps$
frustrated, for some constant $\eps$ independent of $n$, then from any initial string it is possible to reach a bad string in 
$r$ steps, where $r$ is a constant that only depends on $\eps$, $k$ and $d$. Therefore, we can define an NP verification algorithm
which, given some initial string $x$, tries all possible $r$-size paths, and this can be performed in polynomial time since $r$ is constant.
We describe now the main ideas on how to prove that for highly frustrated Hamiltonians, such a short path 
always exists.

Our proof is based on the following two key ideas. First, we notice that if we start with any initial quantum state which is a uniform superposition of good strings, then in case the Hamiltonian is highly frustrated, there must be a term $H_i$ which has large energy with respect to that string (in fact, there must be many, but for now we focus on one). When we apply on the state the projection $P_i$ onto the groundspace of that frustrated term $H_i$, then it is not very difficult to see that the number of strings in the support of the new state, after this projection, will be larger by a constant factor. Moreover, the value of this {\it expansion factor} is directly related to the amount of frustration of $H_i$ with respect to the state we started with. In other words, the more frustrated the term is, the larger the expansion of the set of strings would be, due to projection with respect to that term. 
We call this phenomenon ``one term expansion''; 
it is proven in \Cref{lem:one-term}. 

Now, the idea is to start with one good string provided by the prover, and expand it to an increasingly larger set of good strings by such projections. Our goal is to perform such expansions by a  ``circuit of parallel non-overlapping\footnote{Two
terms are non-overlapping if the sets of qudits on which they act are disjoint}
projections'', as in \Cref{fig:layers-a}. We would like to argue that if the
frustration of the Hamiltonian is high for any state, as we assume now, then there is a constant fraction of the $m$ terms, given by one layer in the circuit, which are all at least constantly frustrated. By the single term expansion argument, each such term would contribute a constant multiplicative factor to the number of good strings in our set, and thus the size of the set of good strings accumulates an exponential factor due to each {\it layer} in the circuit. If this is true, then it must be the case that after at most constantly many layers, the argument breaks down (namely, a bad string is found) since otherwise the number of strings would exceed the number of all possible strings. The implication is that after constantly many layers, a bad string is reached.

Unfortunately, there is a problem in applying the above line of thought directly:  the amount of expansion of two different terms might be strongly correlated. Let us see an example of that.

\begin{example}
\label{ex:parallel}
Let $S = \{0000, 0011, 1100, 1111\}$ and let $\localproj{P}_{1,4} = \localproj{P}_{2,3}  = \kb{\Phi^+} + \kb{\Psi^+}$, where $\localproj{P}_{i,j}$ acts on qubits $i$ and $j$ (and are implicitly tensored with identity on the other qubits), 
and  $\ket{\Phi^+} = \frac{1}{\sqrt{2}}\left(\ket{00} + \ket{11}\right)$ $\ket{\Psi^+} = \frac{1}{\sqrt{2}}\left(\ket{01} + \ket{10}\right)$ are two of the Bell states.
We notice that 
$\bra{S}\localproj{P}_{1,4}\ket{S} = \frac{1}{2}$
and the same holds for $\localproj{P}_{2,3}$. 
 
However, if we take the support of $\localproj{P}_{1,4}\ket{S}$, $S' = \{0000, 0110, 0011, 0101, 1100, 1010, 1111, 1001\}$, it follows that  $\bra{S'}\localproj{P}_{2,3}\ket{S'} = 0$, so $\localproj{P}_{2,3}$ has no frustration after we correct the frustration of $\localproj{P}_{1,4}$. 
\end{example}

This example means that we cannot use the above argument as stated, for many non-overlapping terms applied in parallel: even though there are indeed a linear number of non-overlapping terms which are all frustrated, we cannot simply multiply the expansions due to each of them.

We overcome this difficulty by resorting to some 
``online''  version of the claim: it turns out that by using an adaptive argument, a constant fraction of terms can be found which will all contribute independent multiplicative factors to the increase in size of the set of good strings. This means that each layer in the circuit does contribute an exponential increase in the number of strings. 

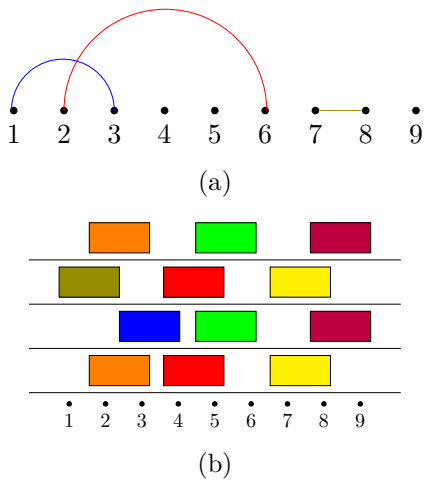
\begin{figure}
    \centering
\begin{subfigure}[b]{0.35\textwidth}
\begin{adjustbox}{max totalsize={1\textwidth}{\textheight},center}
\begin{tikzpicture}
    \node[circle,fill=black,inner sep=0pt,minimum
    size=3pt,label=below:{$1$}] (x1)  {};
  \foreach \i/\j in {2/1,3/2,4/3,5/4,6/5,7/6,8/7, 9/8}
  {
  \node[right=1.5em of x\j,circle,fill=black,inner sep=0pt,minimum
    size=3pt,label=below:{$\i$}] (x\i) {};   
  }

  \draw[color=blue] (x3) arc (0:180:1.85em);
  \draw[color=red] (x2) arc (0:-180:-3.65em);
  \draw[color=olive] (x7) -- (x8);
\end{tikzpicture}
\end{adjustbox}
\caption{}
\label{fig:layers-a}
\end{subfigure}

    \begin{subfigure}[b]{0.3\textwidth}
\begin{adjustbox}{max totalsize={1\textwidth}{\textheight},center}
\begin{tikzpicture}
    \node[circle,fill=black,inner sep=0pt,minimum
    size=3pt,label=below:{$1$}] (x1)  {};
  \foreach \i/\j in {2/1,3/2,4/3,5/4,6/5,7/6,8/7, 9/8}
  {
  \node[right=1.5em of x\j,circle,fill=black,inner sep=0pt,minimum
    size=3pt,label=below:{$\i$}] (x\i) {};   
  }
    \node[above=0.4em of x1] (x1l0)  {};    
  \foreach \i/\j/\k/\ck/\l/\cl/\m/\cm in {1/0/1em/orange/4.7em/red/10em/yellow,2/1/2.5em/blue/6.3em/green/12em/purple,
  3/2/-0.5em/olive/4.7em/red/10em/yellow,
  4/3/1em/orange/6.3em/green/12em/purple}
  {
    \draw (x1l\j) ++(-2em,-0.35em) -- ++(18.5em,0);  
    \node[above=1.5em of x1l\j] (x1l\i)  {};
    \draw[draw=black, fill=\ck] (x1l\j) ++ (\k,0) rectangle ++(3em,1.5em);
    \draw[draw=black, fill=\cl] (x1l\j) ++ (\l,0) rectangle ++(3em,1.5em);
    \draw[draw=black, fill=\cm] (x1l\j) ++ (\m,0) rectangle ++(3em,1.5em);    
}
\end{tikzpicture}
\end{adjustbox}
\caption{}
\label{fig:layers-b}
\end{subfigure}
\caption{In \Cref{fig:layers-a}, we show an example of non-overlapping $2$-local terms, where each term corresponds to an edge of the graph.
In \Cref{fig:layers-b}, we depict a constant-depth circuit where each layer contains non-overlapping $2$-local terms.}
\label{fig:layers}
\end{figure} 

We can now provide a sketch of the first part of the proof: in the case of $\epsilon$-frustration, assume we start with a subset-state of good strings $\ket{S}$, and let $L\ket{S}$ be the state which we arrive at, after applying all non-overlapping projections in the sequence $L$ which we have found above, and taking the subset-state of all strings we have reached. We can show that $L\ket{S}$ contains $\left(1+\frac{\eps}{4}\right)^{\frac{\eps n}{2kd}}$ more strings than  $S$. Then, we repeat this constantly many times. More concretely, set $S_0 = \{x\}$, for some initial string $x$. The above argument shows that either $\ket{S_0}$ contains a bad string (i.e. $x$ is a bad string), or there is a set of terms $L_1$ such that the set $S_1$ with the strings in the support of $L_1\ket{S_0}$  has exponentially more strings than $S_0$. We now repeat this process starting with the state $\ket{S_1}$ instead of $\ket{S_0}$, and so on, until we reach a bad string. Since the number of strings in the set increases exponentially at every step,
there exists some constant $\ell$, that depends only on  $\eps$, $k$ and $d$, such
that  $S_{\ell}$ (which we prove to be the strings in the support of $L_{\ell}...L_1\ket{x}$) contains a bad string. This shows that a constant depth 
circuit of non-overlapping projections, applied to an input string, leads to a bad string.  We depict this constant depth circuit in  \Cref{fig:layers-b}. The proof of the claim that within constantly 
many layers a bad string is reached, is given in \Cref{lem:constant-layers}.

We notice that such a constant depth circuit implies that a bad string can be
found within a constant number of rounds, where each round consists of a set of
local steps, each changing a different local part of the string. However the number of steps in each round might be linear, leading to exponentially many 
possibilities; thus, the
brute-force search for such a path is intractable. 

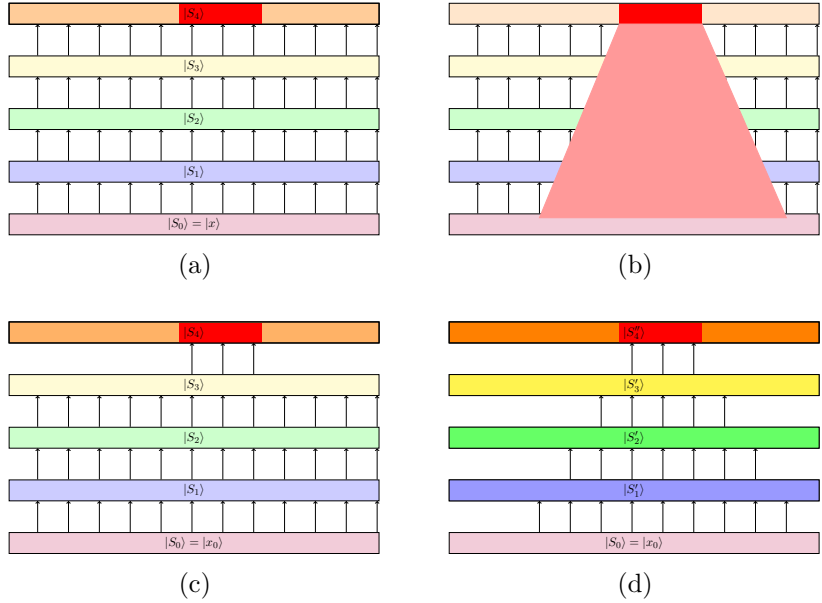
\begin{figure}[t]
    \centering
    \begin{subfigure}[b]{0.3\textwidth}
\begin{adjustbox}{max totalsize={1\textwidth}{\textheight},center}
\begin{tikzpicture}
  \foreach \i/\ip/\c/\sb/\se/\sli in {0/1/blue/2/10/5,1/2/green/3/9/4,2/3/yellow/4/8/3,3/4/orange/5/7/2}
    {
      \ifthenelse{\i=0}{
        \node[draw, thick, rectangle, minimum
          width=30em, minimum height=1.7em, fill=purple!20] (x0) {$\ket{S_0} = \ket{x}$};
      }{
     }
      \foreach \j/\jp in {0/0,1/0,2/1,3/2,4/3,5/4,6/5, 7/6, 8/7, 9/8, 10/9,
      11/10}
      {
         \ifthenelse{\j=0}{
           \node at (x\i.176) (n\i0) {};
         }{
           \node[right=1.8em of  n\i\jp.0] (n\i\j) {};
         }
         \ifthenelse{\j < \sb \OR \j > \se}{
           \draw[->] (n\i\j) ++ (0em, -0em) -- ++(0em,2.5em) {};
         }{
           \draw[->] (n\i\j) ++ (0em, -0em) -- ++(0em,2.5em) {};
         }
      }
      \node[above=2.5em of x\i] (x\ip) [draw, thick, rectangle, minimum
        width=30em, minimum height=1.7em, fill=\c!20] {$\ket{S_{\ip}}$};
    }
      \node (dummy) {};
  \node[above=2.5em of x3] (x42) [draw, thick, rectangle, minimum
        width=30em, minimum height=1.7em, fill=orange!40] {};
  \fill [red] (x4.south) ++ (-1.2em, 0.05em) rectangle ++(6.7em,1.6em);
      \node[above=2.5em of x3] (x4b) [draw, thick, rectangle, minimum
        width=30em, minimum height=1.7em] {$\ket{S_4}$};  
\end{tikzpicture}
\end{adjustbox}
\caption{}
\label{fig:frustration-a}
\end{subfigure}
\;\;\;\;\;\;
\begin{subfigure}[b]{0.3\textwidth}
\begin{adjustbox}{max totalsize={1\textwidth}{\textheight},center}
\begin{tikzpicture}
  \foreach \i/\ip/\c/\sb/\se/\sli in {0/1/blue/2/10/5,1/2/green/3/9/4,2/3/yellow/4/8/3,3/4/orange/5/7/2}
    {
      \ifthenelse{\i=0}{
        \node[draw, thick, rectangle, minimum
          width=30em, minimum height=1.7em, fill=purple!20] (x0) {};
      }{
     }
      \foreach \j/\jp in {0/0,1/0,2/1,3/2,4/3,5/4,6/5, 7/6, 8/7, 9/8, 10/9,
      11/10}
      {
         \ifthenelse{\j=0}{
           \node at (x\i.176) (n\i0) {};
         }{
           \node[right=1.8em of  n\i\jp.0] (n\i\j) {};
         }
         \ifthenelse{\j < \sb \OR \j > \se}{
           \draw[->] (n\i\j) ++ (0em, -0em) -- ++(0em,2.5em) {};
         }{
           \draw[->] (n\i\j) ++ (0em, -0em) -- ++(0em,2.5em) {};
         }
      }
      \node[above=2.5em of x\i] (x\ip) [draw, thick, rectangle, minimum
        width=30em, minimum height=1.7em, fill=\c!20] {};
    }
      \node (dummy) {};
  \fill [red] (x4.south) ++ (-1.2em, 0.05em) rectangle ++(6.7em,1.6em);
    \filldraw[red!40] (x4.south) ++ (-1.2em, 0.05em) -- ++(6.7em,0em)
                             -- (n010.south) -- (n02.south) -- cycle;
                             
\end{tikzpicture}
\end{adjustbox}
\caption{}
\label{fig:frustration-b}
\end{subfigure}
\\~\\
    \begin{subfigure}[b]{0.3\textwidth}
\begin{adjustbox}{max totalsize={1\textwidth}{\textheight},center}    
\begin{tikzpicture}
  \foreach \i/\ip/\c/\sb/\se/\sli in {0/1/blue/2/10/0,1/2/green/3/9/0,2/3/yellow/4/8/0,3/4/orange/5/7/1}
    {
      \ifthenelse{\i=0}{
        \node[draw, thick, rectangle, minimum
          width=30em, minimum height=1.7em, fill=purple!20] (x0) {$\ket{S_0} = \ket{x_0}$};
      }{
     }
      \foreach \j/\jp in {0/0,1/0,2/1,3/2,4/3,5/4,6/5, 7/6, 8/7, 9/8, 10/9,
      11/10}
      {
         \ifthenelse{\j=0}{
           \node at (x\i.176) (n\i0) {};
         }{
           \node[right=1.8em of  n\i\jp.0] (n\i\j) {};
         }
         \ifthenelse{\sli=1}{
           \ifthenelse{\j < \sb \OR \j > \se}{}{
             \draw[->] (n\i\j) ++ (0em, -0em) -- ++(0em,2.5em) {};
           }
         }{
           \draw[->] (n\i\j) ++ (0em, -0em) -- ++(0em,2.5em) {};
         }
      }
      \node[above=2.5em of x\i] (x\ip) [draw, thick, rectangle, minimum
        width=30em, minimum height=1.7em, fill=\c!20] {$\ket{S_{\ip}}$};
    }
      \node (dummy) {};
  \node[above=2.5em of x3] (x42) [draw, thick, rectangle, minimum
        width=30em, minimum height=1.7em, fill=orange!60] {};
  \fill [red] (x4.south) ++ (-1.2em, 0.05em) rectangle ++(6.7em,1.6em);
      \node[above=2.5em of x3] (x4b) [draw, thick, rectangle, minimum
        width=30em, minimum height=1.7em] {$\ket{S_4}$};  
\end{tikzpicture}
\end{adjustbox}
\caption{}
\label{fig:frustration-c}
\end{subfigure}
\;\;\;\;\;\;
    \begin{subfigure}[b]{0.3\textwidth}
\begin{adjustbox}{max totalsize={1\textwidth}{\textheight},center}
\begin{tikzpicture}
  \foreach \i/\ip/\c/\sb/\se/\sli in {0/1/blue/2/10/1,1/2/green/3/9/1,2/3/yellow/4/8/1,3/4/orange/5/7/1}
    {
      \ifthenelse{\i=0}{
        \node[draw, thick, rectangle, minimum
          width=30em, minimum height=1.7em, fill=purple!20] (x0) {$\ket{S_0} = \ket{x_0}$};
      }{
     }
      \foreach \j/\jp in {0/0,1/0,2/1,3/2,4/3,5/4,6/5, 7/6, 8/7, 9/8, 10/9,
      11/10}
      {
         \ifthenelse{\j=0}{
           \node at (x\i.176) (n\i0) {};
         }{
           \node[right=1.8em of  n\i\jp.0] (n\i\j) {};
         }
           \ifthenelse{\j < \sb \OR \j > \se}{}{
             \draw[->] (n\i\j) ++ (0em, -0em) -- ++(0em,2.5em) {};
           }
      }
      \node[above=2.5em of x\i] (x\ip) [draw, thick, rectangle, minimum
        width=30em, minimum height=1.7em, fill=\c!20] {};
    }
      \node (dummy) {};
  \node[above=2.5em of x3] (x42) [draw, thick, rectangle, minimum
        width=30em, minimum height=1.7em, fill=orange!100] {};
  \fill [red] (x4.south) ++ (-1.2em, 0.05em) rectangle ++(6.7em,1.6em);

  \node[above=2.5em of x2] (x42) [draw, thick, rectangle, minimum
        width=30em, minimum height=1.7em, fill=yellow!80] {$\ket{S'_3}$};

  \node[above=2.5em of x1] (x42) [draw, thick, rectangle, minimum
        width=30em, minimum height=1.7em, fill=green!60] {$\ket{S'_2}$};

  \node[above=2.5em of x0] (x42) [draw, thick, rectangle, minimum
        width=30em, minimum height=1.7em, fill=blue!40] {$\ket{S'_1}$};
  \node[above=2.5em of x3] (x42) [draw, thick, rectangle, minimum
        width=30em, minimum height=1.7em] {$\ket{S''_4}$};         
\end{tikzpicture}
\end{adjustbox}
\caption{}
\label{fig:frustration-d}
\end{subfigure}
\caption{Example of the light-cone argument used to find a constant-step path from the initial string to a bad string. The red rectangle marks the qudits of the term frustrated by a bad string. In \Cref{fig:frustration-a}, we have a constant number of layers reaching a state with a bad string.
In \Cref{fig:frustration-b}, we show the light-cone from the frustrated term.
In \Cref{fig:frustration-c}, we remove the projections outside of the light-cone on the last layer. Notice that the new state on the last layer still contains a bad string.
In \Cref{fig:frustration-d}, we show the projections that are left after removing all the terms outside of the light-cone. Again, the last layer still contains a bad string.
}
\label{fig:lightcone}
\end{figure} 

The next part of the proof is where we show how to find a bad string efficiently, given that one is reached in the above constant depth projection circuit. We call this ``the light cone argument". 
To retrieve a constant size path from 
$L_{\ell}...L_1\ket{x}$, the key point is noticing that badness of a string is a
{\em local} property, namely, if a string is bad, we can point at at least one
local term which it is bad for; let us refer to this term as the frustrated term
(this is a meaningful name since if a state contains that bad string, then that term will indeed be frustrated). The crux of the matter is that the fact that badness of a string is local, 
implies also that projections on one set of qubits 
does not affect the badness of terms restricted to the complementary set of
qudits (see \Cref{rem:locality-phi}). This implies, by a simple 
argument, that even if 
we remove all of the terms in $L_1, ..., L_{\ell}$ that are not in the {\it lightcone} of the frustrated term, we will still achieve 
a bad string. This is because if $L_{\ell}...L_1\ket{x}$ contains a bad string, which is bad for some term, then any projection in $L_{\ell}...L_1$ which is not in the light cone of that term, cannot influence its badness. We depict such argument in~\Cref{fig:lightcone}; 
It is proven in \Cref{rem:locality-phi}.

Using the light-cone argument, we deduce that instead of applying the layers $L_{\ell}...L_1$, we can apply just the terms in these layers which are contained in the lightcone, denoted by   
$\lightcone{L}_1$, ..., $\lightcone{L}_{\ell}$, to conclude 
that also the state $\lightcone{L}_{\ell}...\lightcone{L}_1\ket{x}$ contains a bad string. Since the lightcone operators are of constant size, every string in $\lightcone{L}_i...\lightcone{L}_1\ket{x}$ can be reached from a string in $\lightcone{L}_{i-1}...\lightcone{L}_1\ket{x}$ in a constant number of steps, and by induction we deduce that there is a short path from $x$ to a bad string.

\subsection{Related work, implications and open problems}\label{sec:discussion}

\paragraph{PCP for AM vs. PCP for MA.} The PCP theorem for AM proved by Drucker~\cite{Drucker11} and mentioned above, relies strongly on the classical PCP theorem: since the randomness in AM is public, both Prover and Verifier agree on the same Boolean formula, and then they can apply the original PCP theorem with this formula.
In the case of MA, such an argument does not hold since the Prover
does not know the formula that will be tested by the Verifier.

\vspace{-0.8em}
\paragraph{PCP for QMA vs. PCP for MA.} We notice the difference between the quantum and the uniform stoquastic PCP conjecture. It is widely believed that  NP = MA, and
therefore the gapped uniform stoquastic Local Hamiltonian problem is believed to be
both in NP, which we prove in this work, and MA-hard, which is the uniform
stoquastic PCP conjecture (\Cref{conj:pcp}). In this sense, the result that we prove is somehow
``expected'', and it does not change our belief that the stoquastic PCP
conjecture does hold (as it is implied by BPP = P). However, in the fully quantum setting, we believe that NP $\ne$ QMA, and if the
constant promise gap  Local Hamiltonian problem is in NP, then this is a strong
indication  that the quantum PCP conjecture is false\footnote{Indeed, proving
that constant promise gap  Local Hamiltonian problem is in NP is often
considered as disproving the quantum PCP conjecture.}. Our result implies also that
if one expects to prove the quantum PCP conjecture without causing the extra
``side-effect'' of proving NP = MA, the gap-amplification process should not preserve uniform-stoquasticity.

\vspace{-0.8em}
\paragraph{Detectability lemma.} Our setting resembles that of the Detectability lemma (DL)~\cite{AharonovALV09}, a useful tool in quantum Hamiltonian complexity \cite{AharonovALV11,AradLV12,GossetY16,AradLVV2017}. Like in our setting, in the DL setting (see the formulation of \cite{AnshuAV16}) one considers a given local Hamiltonian, where each term is associated with
a local projector on its groundspace. Starting with some state, one considers applying the local projectors one by one (in an arbitrary order). 
Under the assumption that every qudit participates in at most constantly many local terms, this can be viewed as applying all local projections, organized in a constant depth circuit made of local projections - similar to our setting. If the state we start with is the groundstate, and there is no-frustration in the Hamiltonian, then the norm of the state after all these projections of course remains one; this is the easy part. The DL says that if the Hamiltonian is frustrated, then the norm of the state after all these projections will have shrunk by at least some factor; the key in the lemma is to upper bound that factor. When the Hamiltonian is highly frustrated, the DL says that the factor will be a constant strictly less than $1$. This is a strong statement; could it possibly be used to deduce our result? 

A closer look reveals important differences between the two questions. While our goal is to argue containment in NP, and thus we need an efficient classical witness (which we take as an $n$-dit string), the DL requires full knowledge of the quantum state on which the projections are applied, in order to deduce the behavior of the norm. 
We thus do not know how to make any usage of the DL in our setting, though interestingly, 
there seem to be some conceptual connections; In particular, like in the current paper, the proof of the DL relies on arguing that the correlations between the different projections do not matter. 
It would be interesting if more can be said in this direction.

\vspace{-0.8em}
\paragraph{Uniform vs. non-uniform case.} Deriving our results without the uniformity restriction on our stoquastic Hamiltonians seems to be an interesting challenge.  We note that since the uniform case is already MA hard, this will not affect our complexity implications, but seems conceptually 
important, as well as interesting from the complexity of groundstates point of view.
Unfortunately we do not know how to do this, and it is left for future work. The main difference between the uniform and the general case is that in the uniform case, the only source of frustration is the existence of bad strings. When we move to the non-uniform case, the frustration might also appear due to amplitude inconsistency. Let us see a simple example of this.

\medskip \medskip
\noindent {\em Example.}
Let $\eps$ be some small value. Let us consider a one-qubit system, and the Hamiltonian consists in the sum of two terms whose groundstate projectors are 
\[\localproj{P}_1 = \frac{1}{2}\left(\ket{0} + \ket{1}\right)\left(\bra{0} + \bra{1}\right) \textrm{ and } \localproj{P}_2 = \left(\sqrt{1 - \eps}\ket{0} + \sqrt{\eps}\ket{1}\right)\left(\sqrt{1 - \eps}\bra{0} + \sqrt{\eps}\bra{1}\right).\]

We notice that the Hamiltonian is frustrated but there are no bad strings! The source of the frustration is the fact that the first terms requires the amplitude of the groundstate to be the same, but on the other hand the second term pushes them far apart. 

\medskip \medskip

BT~\cite{BravyiT09} deal with this problem in their random walk by assigning weights to the edges, where the weights depend on the Hamiltonian term connecting the two strings. Then, frustration implies that the weight of different paths between two pairs of strings in the support of the groundstate, have different weights (a weight of a path is the product of the weights of the edges).  BT~\cite{BravyiT09} then add extra tests to find these inconsistent paths. At every step of the random walk, the verifier rejects if the weight of the path is larger than one (for this the verifier should start the random walk from the string with maximal amplitude in the groundstate). They prove that if the verifier is provided a string whose amplitude was not maximal or that inconsistent paths exist, then  with high probability a random walk finds a path whose weight is larger than one, leading to rejection.

Interestingly, our proof goes through almost all the way, also for the non-uniform case.
The only problem preventing us from extending the proof to the non-uniform case is the light-cone lemma, 
which does not seem to hold when attempting to detect inconsistencies rather than reachability of bad strings.
A different way to say it is that as far as we can tell, the random-walk proposed by Bravyi and Terhal deals with such cases in a non-local way.  The situation seems to reminiscent of the hardness of sampling from constant depth quantum circuits \cite{TerhalV2004,BravyiGK,CoudronSV18,LeGall18}.

\vspace{-0.8em}
\paragraph{PCP for StoqMA.}
As mentioned previously,  the stoquastic Local Hamiltonian problem where we have to decide if the groundstate energy is below some threshold $\alpha$ or above another threshold $\beta$ is StoqMA-complete for inverse polynomial $\beta - \alpha$. In fact, the StoqMA-completeness holds even for a restricted class of stoquastic Hamiltonians: the  Transverse-field Ising mode~\cite{BravyiH17}. We show in our result that if $\alpha$ is negligible, then the problem is still in NP, and we leave as an open question if the problem is still in NP when $\beta - \alpha$ is constant for an arbitrary $\alpha$.

\vspace{-0.8em}
\paragraph{Adiabatic evolution of Hamiltonians.} Bravyi and Terhal used their random walk for stoquastic Hamiltonians to prove that the adiabatic evolution of frustration-free stoquastic Hamiltonians with inverse polynomial spectral-gap can be performed in randomized polynomial time. 
A major open problem remains to 
extend their result to the general case, in which the frustration-free assumption is relaxed. This would lead to a classical simulation of adiabatic optimization, e.g., of D-Wave type algorithms~\cite{FarhiGGS00,BravyiT09,Hastings13a,Calude15}. We leave as an open question whether our techniques can have any implications in that context. 

\vspace{-0.8em}
\paragraph{A classical version of the problem.} Finally, we note that our work suggests 
a completely classical way to phrase 
the MA-complete problem of Bravyi and Terhal \cite{BravyiT09}. It turns out, as we present more formally in \Cref{app:classical}, that  
{\it uniform} stoquastic local Hamiltonians, 
can be rephrased as a problem which is 
a generalization of constraint satisfaction problems (\problem{CSP}), and which we call here 
Set-Constraints Satisfaction problem, or \problem{SetCSP} in short. 
In \problem{SetCSP}
instead of having constraints that should be satisfied by
an assignment $x \in \{0,1\}^n$, we have ``set-constraints'' that should be
satisfied by some {\it set} of strings $S \subseteq \{0,1\}^n$.
A $k$-local single set-constraint acting on the set $B$ of $k$ bits, 
is a collection of 
non-intersecting sets $T_1,...T_\ell$ of $k$ bit strings. 
Roughly, a set $S$ of $n$-bit strings satisfies such a $k$-local set-constraint if first, the $k$-bit restriction of any string in $S$ to $B$ must be contained in one of the sets. Secondly, 
there is a notion of {\it uniformity}: if an $n$-bit string
$x\in S$ 
restricts to some string in $T_j$ in that set-constraint, 
then by replacing in $x$ the $k$-bits to any other string in 
$T_j$, the string is still in $S$. Namely, all elements of the 
subset $T_j$ appear together (with the same extension 
to the remaining bits) or none of them appears.  
An instance of the \problem{SetCSP} consists of $m$ such $k$-local Set Constraints, and we ask if there is a set of $n$-bit strings that satisfies each of the set constraints, or if all sets of $n$-bit strings are far from satisfying them, namely, frustrated (of course, the notion of being "far" needs to be defined). We show that deciding whether a \problem{SetCSP} instance is satisfiable or inverse polynomially frustrated is MA-complete.  

Our result can be presented in \problem{SetCSP} language, e.g. Theorem \ref{thm:main} implies that for constant promise gap, this problem is actually in NP. However we prefer to present our results here 
using quantum language since the equivalent stoquastic Hamiltonian problem seems to be a more natural problem. 
Never the less, this first complete problem 
for MA defined in CSP language might be useful in future works.  

\paragraph{Organization of the remainder of the paper:}
We start with some preliminaries in \Cref{sec:preliminaries}. We discuss stoquastic Hamiltonians and the proof of MA-completeness in \Cref{sec:stoquastic}. Our main result is proven in \Cref{sec:main}. 
In \Cref{sec:negligible}, we show how to replace the frustration-free condition by allowing the frustration to be negligible.
We finish by proving that commuting stoquastic Hamiltonians are in NP in \Cref{sec:commuting}. We leave to~\Cref{sec:coRP} the definition of the variants of the stoquastic Hamiltonian problems which are hard for co-RP and BPP. In \Cref{app:classical} we provide the description of the uniform stoquastic Hamiltonian problem in CSP language, and the proof it is MA complete. 
 
\subsection*{Acknowledgments}
The authors thank Ayal Green for the helpful discussion at early stages of this work. DA is also grateful for  a very short discussion with Noam Nisan which was as insightful as much as it was short, as well as for very useful remarks from Avi Wigderson, and a question from Umesh Vazirani. We also thank Henry Yuen for the help on improving the clarity of the paper.  
AG is supported by 
ERC Consolidator Grant 615307-QPROGRESS. DA's research on this project was supported by ERC grant 280157 and ISF grant 1721/17.
Part of this work was done when AG was a member of IRIF, Universit\'{e} Paris Diderot, Paris, France, where he was supported by ERC QCC. The authors thank also the French-Israeli Laboratory on Foundations of Computer Science   (FILOFOCS) and Simons collaboration grant number 385590 that allowed AG to visit Hebrew University of Jerusalem.

\section{Preliminaries and Notations}
\label{sec:preliminaries}
\subsection{Complexity classes, NP and MA}
A (promise) problem  $A = (A_{yes}, A_{no})$ consists of two non-intersecting sets 
 $A_{yes}, A_{no} \subseteq \{0,1\}^*$. 
We define now the main complexity classes that are considered in this work. We start by formally defining the well-known class \NP{}. 
\defclass{\NP}{def:NP}{
A problem $A=(\ayes,\ano)$ is in \NP{} if and only if there exist a polynomial
  $p$ and a deterministic algorithm $D$, where $D$
  takes as input a string $x\in\Sigma^*$ and a $p(|x|)$-bit witness $y$ and decides
   on acceptance or rejection of $x$
  such that:
}
{If $x\in\ayes$, then there exists a witness $y$
    such that $D$ accepts $(x,y)$.}
{If $x\in\ano$, then for any witness $y$, $D$
    rejects $(x,y)$.}

We can then generalize this notion, by giving the verification algorithm the power of flip random coins, leading to the complexity class \MA{}.

\defclass{\MA}{def:MA}{
A problem $A=(\ayes,\ano)$ is in \MA{} if and only if there exist a polynomial
  $p$ and 
 a probabilistic algorithm $R$, where $R$
takes as input a string $x\in\Sigma^*$ and a $p(|x|)$-bit witness $y$ and decides
on acceptance or rejection of $x$
such that:  
}
{ If $x\in\ayes$, then there exists a witness $y$
    such that $R$ accepts $(x,y)$ with probability $1$.}
{ If $x\in\ano$, then for any witness $y$, $R$
    accepts $(x,y)$ with probability at most $\frac{1}{3}$.}

The usual definition of  \MA{} requires yes-instances to be accepted with probability at least $\frac{2}{3}$, but it has been shown that there is no change in the computational power  if we
require the verification algorithm to always accept yes-instances~\cite{Zachos1987,Goldreich2011}.

\subsection{Quantum states}
We review now the concepts and notation of Quantum Computation that are used in
this work. We refer to Ref. \cite{NielsenC2011} for a detailed introduction of
these topics.

Let $\Sigma = \{0,...,q-1\}$ be some alphabet.
A qudit of dimension $q$ is associated with the Hilbert space $\complex^{\Sigma}$, whose canonical (also called computational) basis is
$\{\ket{i}\}_{i \in \Sigma}$. 
A pure quantum state of $n$ qudits of dimension $q$ is a unit vector in the Hilbert space
$\left\{\complex^{\Sigma}\right)^{\otimes n}$, where $\otimes$ is the Kroeneker (or tensor) product.
The basis for such Hilbert space is $\{\ket{i}\}_{i \in \Sigma^n}$.  
For some quantum state $\ket{\psi}$, we denote $\bra{\psi}$ as its conjugate transpose. The inner product between two vectors $\ket{\psi}$ and $\ket{\phi}$ is denoted by
$\braket{\psi}{\phi}$ and their outer product as
$\ketbra{\psi}{\phi}$. For a vector $\ket{\psi} \in \complex^{|\Sigma|^n}$, its $2$-norm is defined as ${\norm{\ket{\psi}} :=
  \left(\sum_{i \in \Sigma^n} |\braket{\psi}{i}|^2\right)^\frac{1}{2}}$.

We now introduce some notation which is somewhat less commonly used and more specific for this paper: the support of $\ket{\psi}$,  $supp(\ket{\psi}) = \{i \in \Sigma^n : \braket{\psi}{i} \ne 0\}$, is the set strings with non-zero amplitude.  
We call quantum state $\ket{\psi}$ {\it non-negative} if $\braket{i}{\psi} \geq 0$ for
all $i \in \Sigma^n$.
For any $S \subseteq \Sigma^n$, we define the state $\ket{S} :=
\frac{1}{\sqrt{S}} \sum_{i \in S} \ket{i}$ as the subset-state corresponding to
the set $S$~\cite{Watrous00}. For a non-negative state $\ket{\psi}$, we define
$\ket{\reallywidehat{\psi}} := \ket{supp(\ket{\psi})}$ as the subset-state induced by the strings in the support of $\ket{\psi}$. We say that this is the subset state corresponding to the state 
$\ket{\psi}$. Analogously, for some linear operator $P$ the state  $\reallywidehat{P\ket{\psi}}$ means the subset-state corresponding to the state  $P\ket{\psi}$.

\subsection{Hamiltonians, Groundstates, Energies, Frustration}
\begin{definition}[Hamiltonian]
A {\em Hamiltonian} on $n$ qudits is a Hermitian operator on $\complex^{|\Sigma|^n}$, namely, a complex Hermitian matrix of  dimension $|\Sigma|^n \times  |\Sigma|^n$. 
A Hamiltonian on $n$ qudits is called {\em $k$-Local} if it can be written as $H = \sum_{i = 1}^m \globalproj{H}_i$, where each $\globalproj{H}_i$ can be written in the form $\globalproj{H}_i=H_i\otimes I$, where $H_i$ acts on at most $k$ out of the $n$ qudits.  
\end{definition}

Hamiltonians describe the evolution of physical systems, using Schrodinger's equation. Their eigenvalues correspond to the energy of the system; more generally, the energy of a state $\ket{\psi}$ with respect to a Hamiltonian $H = \frac{1}{m} \sum_{i = 1}^m H_i$ is 
given by $\bra{\psi}H\ket{\psi}$. Notice that we use the term energy 
even though we average by the number of terms, so this is the average energy {\it per term}; this is different from the usual usage of the term energy, or energy density, in the physics literature, where 
one usually considers the average energy {\it per particle}.
This normalization is more convenient in the context of 
PCPs \cite{Dinur07}. 
We also consider the energy of the state with 
respect to a specific term $H_i$, which is $\bra{\psi}H_i\ket{\psi}$.
The minimal energy is the smallest eigenvalue of the Hamiltonian, and an eigenstate which has this energy is called a {\it groundstate}.

\begin{definition}[Groundstate, groundspace, frustration and frustration-free]
A {\em groundstate} of a Hamiltonian $H = \frac{1}{m}\sum_{i = 1}^m H_i$ is an eigenvector associated with its minimum eigenvalue, which is called the {\em groundstate energy}. The {\em groundspace} of a Hamiltonian is the subspace spanned by its groundstates. 
$H$ is called $\eps$-frustrated if for every state $\ket{\psi}$, $\frac{1}{m} \sum_i \bra{\psi}H_i\ket{\psi} \geq \eps$.
Finally, $H$ is called {\em frustration-free} if 
there exists some $\ket{\psi}$ such that 
for every $i$, the local term $H_i$ is positive definite
 and $\bra{\psi}H_i\ket{\psi} = 0$. 
\end{definition}

Throughout this paper we use the following notation for a local Hamiltonian $H = \frac{1}{m} \sum_{i = 1}^m \globalproj{H}_i$.
We set $\localproj{P}_i$ to be the {\it local} projection on
the groundspace of $H_i$; while $\globalproj{P}_i=\localproj{P}_i\otimes I$ corresponds to the projection on the groundspace of $\globalproj{H}_i$. 

We prove now a useful lower-bound on the number of frustrated terms of a highly frustrated Hamiltonian.

\begin{claim}[Lower bound on frustrated terms]\label{lem:lb-frustrated-terms}
  Let $H = \frac{1}{m}\sum_{i=1}^m \tilde{H}_i$ be a Local Hamiltonian that is $\eps$-frustrated. Then for every state $\ket{\psi}$, there exist at least $\frac{\eps m}{2}$ terms that are at least $\frac{\eps}{2}$ frustrated.
\end{claim}
\begin{proof}
We prove this by contradiction. Let $F = \{i : \bra{\psi}\tilde{H}_i\ket{\psi} \geq \frac{\eps}{2}\}$. We assume then that $|F| < \frac{\eps m}{2}$. Then the energy of the state is
\begin{align*}
&\bra{\psi}H\ket{\psi} \\
&= \frac{1}{m}\left(\sum_{i \in F} \bra{\psi}\tilde{H}_j\ket{\psi}  +
\sum_{i :\not\in F} \bra{\psi}\tilde{H}_j\ket{\psi}  \right) \\
&\leq \frac{|F|}{ m} + (m - |F|)\frac{\eps}{2m}   \\
&< \frac{\eps}{2} + \frac{\eps}{2}   \\
&= \eps,
\end{align*}
where in the first inequality we used the fact that the norm of all terms is at most $1$, and that the terms outside of $F$ contribute at most $\frac{\eps}{2m}$ to the above sum by definition of the set $F$. In the second inequality we used our assumption that $|F| < \frac{\eps m}{2}$. We then have that $H$ is not $\eps$-frustrated, which is a contradiction.
\end{proof}

\section{Background: the Stoquastic Hamiltonian problem}
\label{sec:stoquastic}
\label{sec:shamiltonian}
In this section, we define stoquastic Hamiltonians, prove certain basic properties,  as well as state their relation to the complexity class MA. 

\subsection{Stoquastic Hamiltonians}
In this work, we deal with a special type of Hamiltonians, which are called stoquastic.

\begin{definition}[Stoquastic Hamiltonian~\cite{BravyiDOT08}]
 A $k$-Local Hamiltonian $H = \sum_{i = 1}^m \globalproj{H}_i$ is called {\em stoquastic} in the computational basis if for all $i$, the off-diagonal elements of $H_i$ (the local terms) in this basis are non-positive\footnote{Klassen and Terhal~\cite{KlassenT18} have a different nomenclature. They call a matrix Z-symmetric if 
the off-diagonal elements of the local terms are non-positive and they call a Hamiltonian stoquastic if all local terms can be made Z-symmetric by local rotations.
}.
\end{definition}

As remarked in \cite{MarvianLH18}, every Hamiltonian is stoquastic in the basis that diagonalizes it. However, the length of such description might be exponential in the number of qubits, since it may be impossible to write it as a sum of local terms. Some recent works~\cite{MarvianLH18,KlassenT18}, provide evidence that deciding if a given local  Hamiltonian can be made stoquastic by local basis change is computationally hard. Therefore, in our definition we assume that the stoquastic Hamiltonian is {\it given} in the basis where each of the {\it local} terms is stoquastic, i.e., has non-positive off-diagonal elements.

We remark that WLOG we assume that each term $H_i$ is positive semi-definite since we could just add a constant $cI$ to the term which only causes a constant shift in the eigenstates of the total Hamiltonian, hence does not change the nature of the 
problems we discuss here but can be easily seen to make the terms PSD. 

A property of a stoquastic local Hamiltonian is that the groundspace of the local terms can be decomposed in a sum of orthogonal non-negative rank-1 projectors.
\begin{lemma}[Groundspace of stoquastic Hamiltonians, Proposition 4.1 of \cite{BravyiT09}]
\label{lem:gspace-structure}
Let $H$ be a stoquastic Hamiltonian and let $P$ be the projector onto its groundspace. It follows that 
\begin{align}
\label{eq:projector-groundspace}
P = \sum_j \kb{\phi_j}, 
\end{align}
where for all $j$, $\ket{\phi_j}$ is non-negative and for $j \ne j'$,
$\braket{\phi_{j'}}{\phi_j} = 0$.
\end{lemma}
\begin{proof}
We start by showing that if all of the entries of $P$ are non-negative, then the statement holds. 
Let $x, y, z$ be some strings such that $\bra{x}P\ket{y} > 0$ and
$\bra{y}P\ket{z} > 0$. Then
\begin{align}\label{eq:connected}
 \bra{x}P\ket{z} = \bra{x}P^2\ket{z} = \sum_{w} \bra{x}P\kb{w}P\ket{z}
 > \bra{x}P\kb{y}P\ket{z} > 0,
\end{align}
where in the first inequality we use the fact that $P$ has only non-negative entries. Therefore, we can partition the string in equivalent classes $T_1,...T_t$ regarding the property $\bra{x}P\ket{y} > 0$.

It follows that the subspace spanned by the strings in $T_i$ is $P$-invariant and therefore $P$ is block-diagonal with respect to the direct sum of such subspaces. Using the Perron-Frobenius theorem for each of the blocks, we have that its largest eigenvalue is non-degenerate, and in this case the block is rank-one, since $P$ is a projector with eigenvalues $1$ and $0$. Since all the entries of $P$ are non-negative, then each one of these rank-one blocks correspond to a non-negative state.

This finishes the proof of the case where $P$ has non-negative entries. We show now that this property holds for stoquastic Hamiltonians.

We have that the groundspace projector $P$ of the Hamiltonian consists of the Gibbs state for temperature tending to $0$, i.e.,  $P = \lim_{\beta \rightarrow \infty} q \frac{e^{-\beta H}}{Tr(e^{-\beta H})}$, where $q > 0$ is the dimension of the groundspace (This is a well known easy fact, see for example Proposition $4.1$ in  \cite{BravyiT09}). 
Thus, it suffices to prove that $e^{-\beta H}$ is a matrix of non-negative entries (we already know that the trace is non-negative by the fact that the eigenvalues 
of $e^{-\beta H}$ are positive). 

Let $s$ be some value such that $-\beta H + s I$ has only non-negative entries. Write 
\begin{align*}
e^{-\beta H} = e^{-\beta H + s I - s I } = e^{-\beta H + s I}e^{- s I },
\end{align*}
where the last equality holds because $- s I$ and $(-\beta H + s I)$ commute.

Note that the Taylor expansion of $e^A$ is 
\[e^A = \sum_{i = 0}^{\infty} \frac{A^k}{k!}. \]
Thus we have that the entries of $e^{-\beta H + s I}$ are non-negative. Write $e^{-s}=p>0$, and we have 
that all entries of $e^{-\beta H}=pe^{-\beta H+sI}$ are non-negative as well.
\end{proof}

\begin{definition}
(Stoquastic Projector) 
Given a projection matrix $P$ acting on $k$ qudits, if there exists a set of
  orthogonal $k$-qudit non-negative states $\{\ket{\Phi_j}\}_j$,
such that 
\[P=\sum_j \kb{\Phi_j},\]
we say $P$ is a stoquastic projector, and we refer to this unique decomposition 
as a sum of projections to non-negative states 
as the {\it non-negative decomposition} of $P$. 
\end{definition}

\begin{remark}
Notice that \Cref{lem:gspace-structure} implies that 
the projection on the groundspace of a stoquastic Hamiltonian is 
a stoquastic projector. 
\end{remark}

A crucial point in the paper is the fact that when applying 
a local stoquastic projector $P$ on some set of qudits $Q$, 
we do not introduce new strings in the reduced density matrix 
of the set of qudits outside of $Q$. Notice that since we 
are considering non-unitary operators, namely projections, then even though they are local, such projections {\it can} in fact have the effect of {\it removing} strings away from the density matrices of qubits which they do not touch;   
the point of this claim is that they cannot {\it add} new strings 
away from where they act.

\begin{claim}[Local action of projectors]\label{rem:locality-phi}
Let $P$ be a stoquastic projector on a subset $Q$ of $k\le n$ qudits.  
Consider the projection $\globalproj{P}$ on $n$ qudits derived from $P$ by 
$\globalproj{P}=\localproj{P}_Q\otimes I_{\overline{Q}}$. 
Then $\globalproj{P}$ is also a stoquastic projector, 
it can be written as the following non-negative decomposition:  
 \begin{equation}\label{eq:Pi}
 \globalproj{P} = \sum_{z \in \Sigma^{n - k},j} \kb{\phi_{j}}_Q \otimes \kb{z}_{\overline{Q}},\end{equation} 
 where 
  $\{\ket{\phi_{j}}\}_j$ is the non-negative decomposition of 
 $P$, and moreover, for any non-negative state $\ket{\psi}$, we have: 
 \[
 \{x_{\overline{Q}} : x \in supp(\globalproj{P}\ket{\psi})\} \subseteq \{x_{\overline{Q}} : x \in supp(\ket{\psi})\}.\]
\end{claim}
\begin{proof}
  For the first part of the claim, the fact that $\tilde{P}_i=P_i\otimes I$ implies that 
  it can be written in the desired form, 
  and this implies that it is a stoquastic projector. 

For the moreover part, write $\ket{\psi}=\sum_{x\in \Sigma^n} \alpha_x \ket{x}$. We have: 
\begin{align*}
   \globalproj{P}\ket{\psi} =\sum_{x \in supp(\psi)}\sum_{j}\sum_{z\in \Sigma^{n-k}} \alpha_x \ket{\phi_{j}}\braket{\phi_{j}}{x_Q}
   \kb{z}_{\overline{Q}}\ket{x_{\overline{Q}}}_{\overline{Q}}
   =   \sum_{x \in supp(\psi)}\sum_{j} \alpha_x 
   \braket{\phi_{j}}{x_Q} \ket{\phi_{j}}_Q\ket{x_{\overline{Q}}}_{\overline{Q}}
   .
\end{align*}
Let $y \in  supp(\globalproj{P}\ket{\psi})$; so $\bra{y}\globalproj{P}\ket{\psi}\ne 0$. 
By the above expression there must be $x \in supp(\ket{\psi})$ such that $x_{\overline{Q}} = y_{\overline{Q}}$.  
\end{proof}

\begin{remark}[Notation of local and global projectors, and their non-negative decompositions]\label{rem:notation}
As can be seen in \Cref{rem:locality-phi}, we use the tilde to denote the global projector, i.e., $\globalproj{P} = \localproj{P} \otimes I_{\overline{Q}}$. We also extend (in a slightly different way) this notation to the rank-1 projectors and we denote $\ket{\globalproj{\phi}_{j,z}} := \ket{\phi_{j}}_Q\ket{z}_{\overline{Q}}$, for $z \in \Sigma^{n-k}$.
\end{remark}

\begin{remark}[Uniqueness of groundstate containing a string]\label{rem:unique-phi}
Consider a stoquastic projector $\localproj{P}$ and its global version, $\globalproj{P}
 = \sum_{j,z} \kb{\globalproj{\phi}_{j,z}} $(this can be done by
  \Cref{rem:locality-phi} and we use the notation of \Cref{rem:notation}).
Then for every $n$-dit string $x\in \Sigma^n$, there exists at most one
pair of values $j^*, z^*$ such that $\braket{\globalproj{\phi}_{j^*,z^*}}{x} > 0$. 
Clearly, a similar uniqueness statement holds for the local stoquastic projector $\localproj{P}$ and its non-negative decomposition, with respect to $x$ being a $k$-dit string. 
\end{remark}

We now define a particular type of strings, called {\it bad} strings, which play a crucial role in our result. Consider a string $x$ such that $\bra{x}\globalproj{P}_i\ket{x} = 0$; we notice that in this case $x$ cannot belong to any groundstate of $\globalproj{P}_i$. This leads to the following definition: 

\begin{definition}[Bad strings \cite{BravyiT09}]\label{def:bad-string}
  Given a stoquastic projector $\localproj{P}$ on a set $Q$ of $k$ out of $n$ qudits, we say that a string $x \in \Sigma^n$ is {\it bad} for $P$ (or $P$-bad) if $\bra{x}\globalproj{P}\ket{x} = 0$ for $\globalproj{P}=\localproj{P}\otimes I_{\overline{Q}}$. 
  Equivalently, $x$ is bad for $P$, if $\bra{x_Q}\localproj{P}\ket{x_Q}=0$. 
This means that $x_Q$ is not in the support of any of the states in the
  non-negative decomposition, as in \Cref{lem:gspace-structure}, of
  $\localproj{P}$. If a string is not bad for $\localproj{P}$, we say it is
  $\localproj{P}$-good. Given a local Hamiltonian $H=\frac{1}{m}\sum_{i} \tilde{H}_i$, and the
  corresponding stoquastic projectors $\localproj{P}_i$, We say that $x$ is {\it
  bad} for $H$, or $H$-bad, if there exists some $i \in [m]$ such that $x$ is
  bad for $\localproj{P}_i$. 
\end{definition}

One other property that we use is that starting with a non-negative state
$\ket{\psi}$, and applying $\globalproj{P}_i$, the projector onto the groundspace of some local term $H_i$, maintains all the $H_i$-good strings in $\ket{\psi}$. 

\begin{lemma}[Strings added by Local Stoquastic Projectors]
\label{lem:adding strings}
Let $\ket{\psi}$ be a non-negative $n$-qudit state.
Consider $P = \sum_{j} \kb{\phi_{j}}$ a stoquastic projector and its non-negative decomposition, acting on the subset $Q$ of $k$ qudits out 
of these $n$ qudits. Then all $P$-good strings of $\ket{\psi}$ are also in the
  support of  $\globalproj{P}\ket{\psi}$, where
  $\globalproj{P}=\localproj{P}\otimes I_{\overline{Q}}$. Moreover,  it follows that \[supp(\globalproj{P}\ket{\psi}) = \bigcup_{j,z : \braket{\globalproj{\phi}_{j,z}}{\psi} > 0} supp(\ket{\globalproj{\phi}_{j,z}}).\]
\end{lemma}

\begin{proof}
Let $S$ be the support of $\ket{\psi}$ and let also $\ket{\psi} = \sum_{x \in S} \alpha_x\ket{x}$.
Let also $G$ and $B$ be the sets of $P$-good and $P$-bad strings of $n$ dits, respectively. We have that
\begin{align}\label{eq:adding-strings}
   \globalproj{P}\ket{\psi} 
   =   \sum_{x \in S\cap G}\alpha_x \globalproj{P}\ket{x}   + \sum_{x \in S\cap B} \alpha_x \globalproj{P}\ket{x} 
   =   \sum_{x \in S\cap G}\alpha_x \globalproj{P}\ket{x}.  
   \end{align}
We now use \Cref{eq:Pi} from 
  \Cref{rem:locality-phi}, to apply the projector $\globalproj{P}$. We have 
\begin{equation}\label{eq:action}
\globalproj{P}\ket{\psi}   =   \sum_{x \in S\cap G, j, z\in \Sigma^{n-k}} \alpha_x 
\left(\kb{\phi_{j}}_Q\otimes \kb{z}_{\overline{Q}}\right) \ket{x_Q}_{Q}\ket{x_{\overline{Q}}}_{\overline{Q}} =
 \sum_{x \in S\cap G,j} \alpha_x \braket{x_Q}{\phi_{j}} 
\ket{\phi_{j}}_Q  \ket{x_{\overline{Q}}}_{\overline{Q}}
\end{equation}

If $x$ is a $P$-good string,  then there exists a (unique) $\ket{\phi_{j}}$ such that $\braket{\phi_{j}}{x_Q} > 0$. Using 
also that $x\in S$, we see that the amplitude of $x$ in the above expression $\globalproj{P}\ket{\psi}$ is non-zero.

For the moreover part,  notice that \Cref{eq:action} can be written as
\begin{align}\label{eq:action-global}
\globalproj{P}\ket{\psi}   = 
 \sum_{x \in S\cap G,j,z} \alpha_x \braket{x}{\globalproj{\phi}_{j,z}} \ket{\globalproj{\phi}_{j,z}} 
\end{align}
and the statement holds directly.

\end{proof}

\begin{corollary}[Composition of stoquastic projectors]\label{cor:composition}
Let $\globalproj{P}_1$
and $\globalproj{P}_2$
be two $n$-qudit stoquastic projectors (which may or may not be global versions of local projections) and let $\ket{\psi}$ be a non-negative state. Then $\reallywidehat{\globalproj{P}_1 \globalproj{P}_2 \ket{\psi}} = \widehat{\globalproj{P}_1 \widehat{\globalproj{P}_2 \ket{\psi}}} = 
\ket{supp(\globalproj{P}_1\globalproj{P}_2\ket{\psi})}$. Moreover, if $\globalproj{P}_1$ and $\globalproj{P}_2$ commute, then the above states are also equal to 
$\reallywidehat{\globalproj{P}_2 \globalproj{P}_1 \ket{\psi}} = \widehat{\globalproj{P}_2 \widehat{\globalproj{P}_1 \ket{\psi}}}$. 
\end{corollary}
\begin{proof} 
  
  First, we claim that for two non-negative states $\ket{\psi}$ and
  $\ket{\psi'}$, if $supp(\ket{\psi})=supp(\ket{\psi'})$, then 
$supp(\globalproj{P}\ket{\psi})=supp(\globalproj{P}\ket{\psi'})$, for any stoquastic projector 
$\globalproj{P}$.  This can be argued as follows. 
Let $y$ be a string in the support of 
  $\globalproj{P}\ket{\psi}$, i.e. $\bra{y}\globalproj{P}\ket{\psi} \ne 0$, and 
 $\ket{\globalproj{\phi}_{j,z}}$ be the unique state in the 
non-negative decomposition of $\globalproj{P}$ 
that contains $y$, as stated in \Cref{rem:unique-phi}.
Since $\ket{\psi}$ and $\ket{\globalproj{\phi}_{j,z}}$ are both non-negative states, 
we have  that $\bra{y}\globalproj{P}\ket{\psi} = \braket{y}{\globalproj{\phi}_{j,z}}
\braket{\globalproj{\phi}_{j,z}}{\psi} > 0$ and in particular
 $\braket{\globalproj{\phi}_{j,z}}{\psi} > 0$. 
  Notice that since $\ket{\psi'}$ is also a non-negative state and its support
  is equal to the support of $\ket{\psi}$, we also have that 
 $\braket{\globalproj{\phi}_{j,z}}{\psi'} > 0$ and thus
$\bra{y}\globalproj{P}\ket{\psi'} = 
  \braket{y}{\globalproj{\phi}_{j,z}}\braket{\globalproj{\phi}_{j,z}}{\psi'} >
  0$, i.e. $y$ is in the support of $\globalproj{P}\ket{\psi'}$.
  The converse follows from a similar argument.
  
In particular, by definition, $supp(\globalproj{P}_2\ket{\psi}) = supp(\reallywidehat{\globalproj{P}_2\ket{\psi}})$; and thus applying $\globalproj{P}_1$ on both states, we get $supp(\globalproj{P}_1\globalproj{P}_2\ket{\psi}) = supp(\globalproj{P}_1\widehat{\globalproj{P}_2\ket{\psi}})$. 
By definition, we also have 
$supp(\globalproj{P}_1\globalproj{P}_2\ket{\psi})=
supp(\reallywidehat{\globalproj{P}_1\globalproj{P}_2\ket{\psi}})$, which proves the first part 
of the Corollary. 

For the moreover part, we can apply the previous argument to show 
$\reallywidehat{\globalproj{P}_2 \globalproj{P}_1 \ket{\psi}} = \widehat{\globalproj{P}_2 \widehat{\globalproj{P}_1 \ket{\psi}}}
=supp(\reallywidehat{\globalproj{P}_2\globalproj{P}_1\ket{\psi}})$. 
If $\globalproj{P}_1$ and $\globalproj{P}_2$ commute, we have 
$supp(\reallywidehat{\globalproj{P}_1\globalproj{P}_2\ket{\psi}})=supp(\widehat{\globalproj{P}_2\globalproj{P}_1\ket{\psi}})$, 
and the result follows.  
\end{proof}

\subsection{Uniform Stoquastic Hamiltonians}

In this work, we focus on a restricted class of stoquastic Hamiltonian which we
call uniform stoquastic Hamiltonian.

\begin{definition}[uniform stoquastic Local Hamiltonian]
A stoquastic Local Hamiltonian $H = \frac{1}{m} \sum_{i = 1}^m \tilde{H}_i$ is called uniform if the states of the unique non-negative decompositions of each  local stoquastic projector ($\localproj{P}_i$) are subset-states. 
\end{definition}

Following \Cref{rem:locality-phi,rem:notation}, for uniform stoquastic Local Hamiltonians, the groundspace projector of $\tilde{H}_i$ is $\globalproj{P}_i = (\localproj{P}_i)_{Q} \otimes I_{\overline{Q}} = \sum_{j,x} \kb{\localproj{T}_{i,j}} \otimes \kb{x} $, with $\localproj{T}_{i,j} \subseteq \Sigma^k$ and $\localproj{T}_{i,j} \cap \localproj{T}_{i,j'} = \emptyset$ for $j \ne j'$. We also denote $\ket{\globalproj{T}_{i,j,x}} := \ket{\localproj{T}_{i,j}}\ket{x}$ (and $\globalproj{T}_{i,j,x}$ as the corresponding set of $n$-bit strings).

We provide now a lemma which we do not strictly use in the proof, regarding stoquastic frustration-free Local Hamiltonians; that we can 
always assume that the groundstate of the entire 
Hamiltonian is a subset-state.
Though the claim itself is not used, 
it is helpful to conceptually 
hold it in mind, when reading the proof.  

\begin{lemma}[The structure of groundstates of uniform stoquastic Hamiltonian]
Let $H$ be a uniform stoquastic  frustration-free Local Hamiltonian. Let $H_i$ be a local term of $H$. 
Then if $\ket{\psi}$ is a groundstate of $H$, it can be written in the form  
 $\ket{\psi} = \sum_{j,z} \alpha_{i,j,z} \ket{\globalproj{T}_{i,j,z}}$, for some
  choice of coefficients $\alpha_{i,j,z} \in \complex$.  
Moreover, the subset-state $\ket{S}$, for $S= \bigcup_{j,z : \alpha_{j,z} \ne 0} \globalproj{T}_{i,j,z}$, is also a groundstate of $H$. 
\end{lemma}

\begin{proof}
The first claim just follows from the fact that 
$H$ is frustration free so any groundstate must be spanned by groundstates of a fixed term $H_i$, namely \begin{align}\label{eq:structure-gstate-span}
\ket{\psi} = \sum_{j,z} \alpha_{i,j,z} \ket{\globalproj{T}_{i,j,z}}.
\end{align} 

We show now that $\ket{S}$ is also a ground-state of $H$.
Let us consider some term $H_{i'}$ and the decomposition of $\ket{\psi}$ from \Cref{eq:structure-gstate-span} in respect to its non-negative decomposition. It follows  that
  $\bigcup_{j,z :
  \alpha_{i',j,z} \ne 0} \globalproj{T}_{i',j,z} = S$.  
 This implies that
\begin{align*}
   \ket{S} = \sum_{j,z : \alpha_{i',j,z} \ne 0} 
   \frac{\sqrt{\globalproj{T}_{i',j,z}}}{\sqrt{S}} \ket{\globalproj{T}_{i',j,z}},
 \end{align*}
 and therefore $\ket{S}$ is in the groundspace of $H_{i'}$, for any $i' \in [m]$. 
\end{proof}

We formally define now the frustration-free Uniform Stoquastic $k$-Local Hamiltonian problem.

\defproblem{uniform stoquastic frustration-free  $k$-Local Hamiltonian problem}{def:local-hamiltonian}{
  The {\em uniform stoquastic  frustration-free $k$-Local Hamiltonian} problem,
  where $k \in \mathbb{N}^*$ is called the locality and $\eps :  \mathbb{N} \rightarrow [0,1]$ is a non-decreasing function,
  is the following promise problem. Let $n$ be the number of qudits of a quantum system.
  The input is a set of $m(n)$ uniform stoquastic Hamiltonians $H_1, \ldots, H_{m(n)}$
  where $m$ is a polynomial, $\forall i \in m(n) : 0 \leq H_i \leq I$
  and each $H_i$ acts on $k$ qudits out of the $n$ qudit system. We also assume that there are at most $d$ terms that act non-trivially on each qudit, for some constant $d$, and that $m\ge n$.
  For $H = \frac{1}{m(n)} \sum_{i = 1}^{m(n)} H_i$ , one of the following two conditions hold.
}
{There exists a $n$-qudit quantum state
       $\ket{\psi}$ such that
      $\bra{\psi} H \ket{\psi}
        =0$}
{For all $n$-qudit quantum states $\ket{\psi}$
      it holds that
      $\bra{\psi} H \ket{\psi}
        \geq \eps(n) .$
        }

\begin{remark}\label{rem:term-projector}
At some points in this work, we assume that $H_i$ is a projector. This holds without loss of generality, since we could replace an arbitrary Hamiltonian $H_i$ by $H'_i = I - P_i$, where $P_i$ is the projector onto the space of $H_i$, and this could just change frustration by a constant factor.
\end{remark}

\subsection{MA-completeness}
Bravyi and Terhal~\cite{BravyiT09} 
showed that there exists some polynomial $p(n) = O(n^2)$ such that the frustration-free uniform stoquastic $6$-Local Hamiltonian problem with $\eps(n) = \frac{1}{p(n)}$ is MA-hard.
They also  proved that for every polynomial $p'$ and every constant $k$,  the frustration-free uniform stoquastic $k$-Local Hamiltonian problem with $\eps(n) = \frac{1}{p'(n)}$ is in MA. 

Let us start with the direction which is of less technical interest to us, and thus we will not need to go into details. 
The MA-hardness is proved by analyzing the quantum Cook-Levin theorem~\cite{KitaevSV02,AharonovN02}  when
considering an MA verifier. A verification circuit for MA can be described as a quantum
circuit consisting only of the (classical) gates from the universal (classical) gateset 
Toffoli and NOT, operating on
input qubits in the state $\ket{0}$ (the NOT 
gates can then fix them to the right input)
and ancillas which are either in the state $\ket{0}$ used as workspace, or in the state $\ket{+}$, used as random bits. At the end of the circuit, the first qubit is measured in the computational basis and the input is accepted iff the output is $1$.  It is not difficult to check that for such gateset and ancillas, the stoquastic Local Hamiltonian resulting from the circuit-to-Hamiltonian construction of the quantum Cook-Levin theorem (which forces both the correct propagation as well as the correct input state, as well as the output qubit accepting), is a uniform stoquastic Hamiltonian; in particular all the entries are in $\{0, \pm 1, -\frac{1}{2}\}$. We can also assume that each qubit is used in at most $d$ gates, for some $d \geq 3$. This is true because all the computation done by the verifier is classical and therefore the information can be copied to fresh ancilla bits (initialized on $\ket{0}$) with a CNOT operation. Notice then that each qubit takes place on at most $3$ steps: as the target of the CNOT, in some actual computation, and as the source of the next CNOT.

We now explain the other direction, namely Bravyi and Terhal's approach for showing that the stoquastic Hamiltonian problem is in MA. We actually show a simplified version of their result, since we are only
interested in {\it uniform} stoquastic Hamiltonian. Notice that by our above description, this problem is sufficient to achieve MA-hardness. 

Following  \cite{BravyiT09} We now define the graph on which the random walk will take place; this graph is based on a given uniform stoquastic Hamiltonian. 

\begin{definition}[Graph from uniform stoquastic Hamiltonian]
  Let $H = \frac{1}{m}\sum_{i} H_i$  be a uniform stoquastic Hamiltonian on $n$ qudits of dimension $|\Sigma|$. We define the undirected graph $G(H) = (|\Sigma|^n, E)$ where $(x, y) \in E$ iff
there exists a local term $H_i$ with corresponding groundstate projector $\localproj{P}_i$
such that 
\begin{align}
\label{eq:neighbor}
\bra{x}\localproj{P}_i\ket{y} > 0.
\end{align}
\end{definition}

From \Cref{rem:unique-phi,rem:locality-phi}, we have that for a fixed $i$, the neighbor strings form an equivalence class and in each class the strings differ only in the positions where $H_i$ acts non-trivially.
We also remark that given some string $x$, one can compute in polynomial time if $x$ is bad for $H$, by just inspecting the groundspace  of each local term.

The random walk starts from some initial string $x_0$ sent by the prover. If $x_0$ is bad for $H$, then the algorithm rejects. 
Otherwise, a term $H_i$ is picked uniformly at random and a string $x_1$ is
picked uniformly at random from $\ket{\globalproj{T}_{i,j,z}}$, which is the
unique rank-one subset-state from $\globalproj{P}_i$ such that
$x_0 \in \globalproj{T}_{i,j,z}$ (see \Cref{rem:unique-phi}).  The random walk proceeds by repeating this process with $x_1$. We describe
the random walk proposed by BT (simplified for the uniform case)
in~\Cref{algo:bt}.

\begin{figure}[H]
\rule[1ex]{\textwidth}{0.5pt}
\vspace{-20pt}
\begin{enumerate}
 \item Let $x_0$ be the initial string.
 \item Repeat for $T$ steps
 \begin{enumerate} 
   \item If $x_t$ is bad for $H$, reject
   \item Pick $i \in [m]$ uniformly at random
   \item Pick $x_{t+1}$ uniformly at random from the strings in the unique $\globalproj{T}_{i,j,z}$ that contains $x_t$

 \end{enumerate}
 \item Accept
\end{enumerate}
\rule[2ex]{\textwidth}{0.5pt}\vspace{-.5cm}
\caption{BT Random Walk}
\label{algo:bt}
\end{figure}

We state now the lemmas proved in \cite{BravyiT09}.

\begin{lemma}[Completeness, adapted from Section $6.1$ of \cite{BravyiT09}]\label{lem:bt-completeness}
If $H$ is frustration-free, then there exists some string $x$ such that there are no bad-strings in the connected component of $x$.
\end{lemma}

The proof goes by showing that if $H$ is frustration-free, then for any string $x$ in some groundstate of $H$, the uniform superposition of the connected component of $x$ is a groundstate of $H$.
In this case, since all strings in the connected component of $x$ are good (this is by definition of the 
connected component), the verifier will accept. 

\begin{lemma}[Soundness, adapted from Section $6.2$ of \cite{BravyiT09}]
For every polynomial $p$, there exists some polynomial $q$ such that if $H$ is at least $\frac{1}{p(n)}$-frustrated, then for every string $x$, for $T=q(n)$, the random walk from \Cref{algo:bt} rejects with constant probability.
\end{lemma}

The intuition of the proof is that since the Hamiltonian is frustrated, one can upper bound the expansion on any set of good-strings by $1 - \frac{1}{p(n)}$, otherwise the Hamiltonian would not be $\frac{1}{p(n)}$-frustrated. In this case, there exists some polynomial $q$ such that a random walk with $q(n)$ steps escape of any set of good strings with high probability.

\section{Uniform Gapped stoquastic Hamiltonians are in NP}
\label{sec:main}
 
Our main technical result in this work is showing that if a stoquastic uniform
Hamiltonian is $\eps$-frustrated for some constant $\eps$, then every string $x$
is constantly-close to a bad string.

\newcommand{\bodyconstantsteppath}{If the stoquastic uniform $k$-Local
Hamiltonian $H$ is $\eps$-frustrated, then for every string $x$, there is a bad
string $y$ such that the distance between $x$ and $y$ in $G(H)$ is at most  $\sizepath{}$.}
\begin{lemma}[Short path to a bad string]
\label{lem:constant-step-path}
\bodyconstantsteppath
\end{lemma}

Using this lemma, we can prove our main result:

\newtheorem*{thm:repeat-main}{\Cref{thm:main} (restated)}
\begin{thm:repeat-main}
\bodymainthm
\end{thm:repeat-main}
\begin{proof}
The \NP{} witness for the problem consists in some initial string $x$ that is promised to be in the support of the groundstate of $H$. The verification proceeds by running over all
possible $\sizepath{}$-step paths from $x$. Since for each one of the $m$ terms there are constantly 
many possible steps, the number of 
possibilities for one step is polynomial, and so the number constantly-long paths is also polynomial. Therefore, such enumeration can be performed efficiently. For each path, we check if one of the strings it reaches is bad - again this can be done in polynomial time since badness is with respect to the local terms (see \Cref{rem:precision} for the precision issues). The verifier rejects if any of the paths reached a bad string, otherwise it accepts.

Let $x$ be the string sent by the Prover.
If $H$ is frustration-free, then by \Cref{lem:bt-completeness} all strings in the connected component of $x$ are good. On the other hand, if $H$ is $\eps$-frustrated, then  by \Cref{lem:constant-step-path}, there exists a \sizepath{}-step path from $x$ to some string $y$ that is bad for $H$, and such path will be found by the brute-force search.
\end{proof}

\begin{remark}[Deciding on badness of a string
with respect to a local uniform stoquastic term]\label{rem:precision}
We note that while in the non-uniform case the 
question of whether a string is bad for a local term or not, may depend on precision issues, 
this is not a problem when considering uniform stoquastic Hamiltonians. 
In the uniform case, 
the set of strings comprising the subset states 
in the non-negative decomposition of every projector, as in \Cref{eq:projector-groundspace}, can be calculated {\it exactly} given the matrix description of the local Hamiltonian term 
(even if we need to apply approximations when computing the groundstates). 
This is because the locality of the 
Hamiltonian, together with uniformity, imply that if a string is in the support of one of the groundstates, its weight must be $\frac{1}{\sqrt{q}}$ for some positive integer $q$ smaller than some constant. 
\end{remark} 

The remainder of this section is dedicated to proving \Cref{lem:constant-step-path}.
\Cref{sec:expansion} gives the one term expansion argument, 
\Cref{sec:constant-layers} provides the proof that a constant number of layers consisting of parallel non-overlapping projections suffices to reach a bad string; and \Cref{sec:lightcone} provides the light-cone argument to show that if a bad string is reached within constantly many layers, then in fact we there is a bad string within {\it constantly} many steps from the initial string. This then  allows searching for such a string by brute-force. Finally, \Cref{sec:constant-path} just puts all the pieces together to finish the proof of \Cref{lem:constant-step-path}.

\subsection{Expansion}\label{sec:expansion}
We start by showing that if subset-state $\ket{S}$ does not contain any bad string but a 
term $P$ is highly frustrated by $\ket{S}$, then the support of $\globalproj{P}\ket{S}$ is larger than that of $\ket{S}$ by a constant factor.  

\begin{lemma}[One term expansion]
  \label{lem:one-term}
  Let $\globalproj{P} = \sum_{j,z} \kb{\globalproj{T}_{j,z}}$ be a uniform stoquastic projector on the set $Q$ of $k$ out of 
  $n$ qudits. Let $S \subseteq \Sigma^n$ be such that 
  $\ket{S}$ does not contain $P$-bad-strings, and 
  $\norm{\globalproj{P}\ket{S}}^2 \leq 1-\delta$. 
  It follows that the size of the support of 
  $\globalproj{P}\ket{S}$ is at least $(1 + \frac{\delta}{2})|S|$.
\end{lemma}
\begin{proof}
Since $S$ does not contain bad strings, we start by noticing that from \Cref{lem:adding strings}, $S$ is contained in the support of $\globalproj{P}\ket{S}$, and $\globalproj{P}$ only adds the neighbors of strings in $\ket{S}$. 

We have that
 \begin{align}
 1-\delta \geq \|\globalproj{P}\ket{S}\|^2=\bra{S}\globalproj{P}\ket{S}= \sum_{j,z} \braket{S}{\globalproj{T}_{j,z}}\braket{\globalproj{T}_{j,z}}{S} =  
  \sum_{j,z} |\braket{S}{\globalproj{T}_{j,z}}|^2. \label{eq:lb-frustration}
  \end{align}

Let  $\textbf{T} = \bigcup_{j,z :  S \cap \globalproj{T}_{j,z} \ne \emptyset} \globalproj{T}_{j,z}$. 
  It follows that
 \begin{align}
\sum_{j,z} |\braket{S}{\globalproj{T}_{j,z}}|^2 &= 
\sum_{j,z}  \frac{|\globalproj{T}_{j,z} \cap S|^2}{|\globalproj{T}_{j,z}||S|} \nonumber
\\
   &= \sum_{j,z}  \frac{(|\globalproj{T}_{j,z}| - |\globalproj{T}_{j,z} \setminus S|)^2}{|\globalproj{T}_{j,z}||S|} \nonumber\\
   &= \sum_{j,z: \globalproj{T}_{j,z} \cap S \ne \emptyset}  \frac{|\globalproj{T}_{j,z}|^2 -
   2|\globalproj{T}_{j,z}||\globalproj{T}_{j,z} \setminus S|+|\globalproj{T}_{j,z} \setminus S|^2}{|\globalproj{T}_{j,z}||S|} \nonumber \\
   &\geq \frac{|\textbf{T}| -
   2|\textbf{T} \setminus S|}{|S|} \nonumber \\
   &\geq 1  - \frac{2|\textbf{T} \setminus S|}{|S|}. \label{eq:lb-new-strings}   
\end{align}
where in first inequality we remove some non-negative terms and use the fact that $\globalproj{T}_{j,z}$ and $\globalproj{T}_{j',z}$ are disjoint for $j \ne j'$ (\Cref{rem:unique-phi}) and in the second inequality we use the fact that $S \subseteq \textbf{T}$ since there are no bad strings in $S$.

By putting together 
\Cref{eq:lb-frustration,eq:lb-new-strings}, and noticing that 
$\textbf{T} = supp(\globalproj{P}\ket{S})$ from \Cref{lem:adding strings},
we have that
\[
|supp(\globalproj{P}\ket{S})| = |\textbf{T}| = |S| + |\textbf{T} \setminus S| \geq |S| + \frac{\delta}{2}|S| = \left(1 + \frac{\delta}{2}\right)|S|. \qedhere
\]
\end{proof}

\subsection{Bad string in a constant number of layers}\label{sec:constant-layers}

We prove in this section that with a constant number of ``layers'', it is possible to reach a state with a bad string.

We first want find a linear number of non-overlapping terms that are (roughly) simultaneously frustrated by some subset-state $\ket{S}$. Let us first define what we mean by non-overlapping terms.

\begin{definition}[Non-overlapping projectors]
A sequence of local projectors $L = (\localproj{Q}_1,...,\localproj{Q}_{\ell})$
  is non-overlapping if for any $i \ne j$, $\localproj{Q}_i$ and
  $\localproj{Q}_j$ act on disjoint sets of qudits.
\end{definition}

Now, we need to be more careful in order to define the notion of ``simultaneous''. Recall that if $H$ is $\eps$-frustrated, by \Cref{lem:lb-frustrated-terms}, there must exist at least $\frac{\eps m}{2}$ terms that are at least $\frac{\eps}{2}$-frustrated. However, as we explained in \Cref{ex:parallel}, their frustration may be correlated due to entanglement, and when we ``correct'' the frustration of one term, we could also be correcting the frustration of other terms, 

Because of this, we need to choose the sequence of projectors more carefully. We are looking for projectors which are frustrated, even after applying the previous projectors in the sequence.

\begin{definition}[Sequentially frustrated terms]
A sequence of projectors $L = (\localproj{Q}_1,...,\localproj{Q}_{\ell})$ is {\em sequentially}
  $\delta$-frustrated by some subset-state $\ket{S_0}$, if for all
  $i=1,...,\ell$, $\norm{\globalproj{Q}_i \reallywidehat{\globalproj{Q}_{i-1}...\globalproj{Q}_1\ket{S_0}}}^2 \leq 1 - \delta$.
\end{definition}

We show now that we can find a linear-size sequence that is non-overlapping and sequentially highly frustrated, in a greedy way. At iteration $i$, we fix a projector $\localproj{Q}_i$ such that $\globalproj{Q}_i$ is highly frustrated by  $\reallywidehat{\globalproj{Q}_{i-1}...\globalproj{Q}_1\ket{S_{0}}}$ and $\localproj{Q}_i$ does not overlap with any $\localproj{Q}_j$ for $j < i$. More concretely, we choose $\localproj{Q}_i$ arbitrarily from the intersection of the following sets:
\begin{itemize}
\item $F_i$, the set of terms that are at least $\frac{\eps}{2}$-\textbf{F}rustrated by $\reallywidehat{\globalproj{Q}_{i-1}...\globalproj{Q}_1\ket{S_{0}}}$%
\item $A_i$, the set of \textbf{A}vailable terms, i.e. the terms that do not overlap with $\localproj{Q}_{j}$, for $j < i$.
\end{itemize}

We describe such an algorithm in \Cref{algo:find-frustrated-terms} and analyze its correctness in \Cref{lem:non-overlapping-expanding-terms}.

\begin{figure}[H]
\rule[1ex]{\textwidth}{0.5pt}
\vspace{-20pt}

  Let $H = \frac{1}{m} \sum_{j=1}^m H_j$ be a stoquastic $k$-Local Hamiltonian with frustration at least
  $\eps$ and some subset-state $\ket{S_0}$. For each term $H_j$, we denote by $\localproj{P}_j$ the projector onto its groundspace. 

\begin{enumerate}
 \item Let $i = 0$, $A_0 = \{\localproj{P}_1,...,\localproj{P}_m\}$ and $F_0 =  \{\localproj{P}_j : \norm{\globalproj{P}_j \ket{S_{0}}}^2 \leq 1 - \frac{\eps}{2}\}$
 \item While $A_i \cap F_i \ne \emptyset$
 \begin{enumerate} 
   \item Pick any $\localproj{P}_j \in A_i \cap F_i$ and set $\localproj{Q}_i  = \localproj{P}_j$
   \item Let $A_{i+1} = A_i \setminus \{\localproj{P}_j : \localproj{P}_j \textrm{ overlaps with } \localproj{Q}_i\}$ and $F_{i+1} =  \{\localproj{P}_j : \norm{\globalproj{P}_j \reallywidehat{\globalproj{Q}_i...\globalproj{Q}_0\ket{S_{0}}}}^2 \leq 1 - \frac{\eps}{2}\}$
   \item Let $i = i +1$
 \end{enumerate}
 \item Output $L = (\localproj{Q}_0, ..., \localproj{Q}_{i-1})$
\end{enumerate}
\rule[2ex]{\textwidth}{0.5pt}\vspace{-.5cm}
\caption{Algorithm for finding non-overlapping frustrated terms}
\label{algo:find-frustrated-terms}
\end{figure}

\begin{lemma}[Linear number of sequential non-overlapping frustrated terms]
\label{lem:non-overlapping-expanding-terms}
 Let $H = \frac{1}{m}\sum_{i = 1}^m H_i$ be a $\eps$-frustrated uniform stoquastic $k$-Local Hamiltonian, $S_0 \subseteq \Sigma^n$, and $L = (\localproj{Q}_0, ..., \localproj{Q}_{i^*-1})$ be the output of \Cref{algo:find-frustrated-terms}. Then $i)$ $L$ is non-overlapping, $ii)$ $L$ is sequentially $\frac{\eps}{2}$-frustrated by $\ket{S_0}$, and $iii)$ $i^* \geq  \frac{\eps n}{2kd}$.
\end{lemma}
\begin{proof}
Properties $i)$ and $ii)$ follow by construction: $\localproj{Q}_i \in A_i$, and thus it does not overlap with $\localproj{Q}_{j}$ for $j < i$; and $\localproj{Q}_i \in F_i$, therefore $\norm{\globalproj{Q}_i\reallywidehat{\globalproj{Q}_{i-1}..\globalproj{Q}_0\ket{S_{0}}}}^2 \leq 1 - \frac{\eps}{2}$.

We prove now property $iii)$. Notice that $|A_{i+1}| - |A_{i}| \leq kd$, since the only difference between these two sets are the overlapping terms of $\localproj{Q}_i$, $\localproj{Q}_i$ acts on at most $k$ qudits, and there are at most $d$ other terms that overlap with $\localproj{Q}_i$ due to a specific qudit. Therefore, we have that  
 \begin{align}\label{eq:lb-aj}
   |A_{i^*}| \geq m- i^*kd.
  \end{align}
 
Notice that for every $i$, we have that $|F_{i}| \geq \frac{\eps m}{2}$ by \Cref{lem:lb-frustrated-terms}.
We also have that if  $|A_{i}| + |F_{i}| > m$, then $A_{i} \cap F_{i} \ne \emptyset$  by the pigeonhole principle.  Therefore, $A_{i^*}\cap F_{i^*} = \emptyset$ implies that
 \begin{align}
   \left(1  - \frac{\eps}{2}\right) m \geq  |A_{i^*}|  \label{eq:ub-aj}
  \end{align}
 
 Putting \Cref{eq:lb-aj,eq:ub-aj} together
we have
  \[ m - i^*kd \leq |A_{i^*}| \leq  \left(1  - \frac{\eps}{2}\right) m \]
  and therefore it follows that 
  \[  i^*\geq \frac{\eps m}{2 kd} \geq \frac{\eps n}{2 kd}, \]  
where we use the fact that $m \ge n$.
\end{proof}

\begin{definition}[A layer acting on $\ket{S_0}$]\label{rem:commuting} We denote 
$L\ket{S_0}$ to be $\globalproj{Q}_{i^*-1}...\globalproj{Q}_0\ket{S_0}$. We notice that this state is equal to  $\globalproj{Q}_{\sigma(i^*-1)}...\globalproj{Q}_{\sigma(0)}\ket{S_0}$ for every permutation $\sigma$ on the indices $0,...,i^*-1$, since by property i) of \Cref{lem:non-overlapping-expanding-terms} the terms in L are non overlapping, and thus commuting. %
Thus the resulting state does not depend on the order of application of these projections, hence the notion of {\it applying a layer of non-overlapping terms on a state} is well defined. 
\end{definition}

By combining \Cref{lem:one-term,lem:non-overlapping-expanding-terms} we can now prove that if we apply all the projections output by \Cref{algo:find-frustrated-terms}, namely all projections in $L$, the resulting state has exponentially more strings than the original one.

\begin{corollary}[Multiple terms expansion]
\label{lem:expansion}
  Let  $H = \frac{1}{m}\sum_{i = 1}^m H_i$ be an $\eps$-frustrated uniform stoquastic $k$-Local Hamiltonian and $\ket{S}$ be a subset-state such that $\ket{S}$ does not contain any bad string for $H$. Then there is a sequence $L$ of  non-overlapping terms of $H$ such that the number of strings in the support of $L\ket{S}$ is at least  $(1 + \frac{\eps}{4})^{\frac{\eps n}{2 k d}}$ times the number of strings in $\ket{S}$.
\end{corollary}
\begin{proof}
Let $L = (\localproj{Q}_1,...,\localproj{Q}_{i^*})$ be the output of \Cref{algo:find-frustrated-terms}. Let us argue now that
we can use \Cref{lem:one-term} for each one of the $i^* \geq \frac{\eps n}{2 k d}$ terms in $L$  sequentially, which would imply the statement. 

First, let us claim that for all $j$, $\globalproj{Q}_{j}...\globalproj{Q}_1\ket{S}$ does not contain $\localproj{Q}_{j'}$-bad strings, for $j'>j$.  
This follows by induction; for $j=0$ this is true by assumption. 
Also from \Cref{lem:non-overlapping-expanding-terms}, the terms in $L$ are
  non-overlapping. Hence, by  \Cref{rem:locality-phi} when we apply $\globalproj{Q}_j$ no $\localproj{Q}_{j'}$-bad strings appear, for all $j' \ne j$.  
Secondly, directly from \Cref{lem:non-overlapping-expanding-terms} we have that $\norm{\globalproj{Q}_{j+1}\reallywidehat{\globalproj{Q}_{j}...\globalproj{Q}_1\ket{S}}} \leq 1 - \frac{\eps}{2}$.

This means that the conditions of   \Cref{lem:one-term} are satisfied, and it follows that the number of strings in the support of 
the state $\globalproj{Q}_{j+1}\reallywidehat{\globalproj{Q}_{j}...\globalproj{Q}_1\ket{S}}$
 is larger by a factor of $(1 + \frac{\eps}{4})$ than that of the state 
$\reallywidehat{\globalproj{Q}_{j}...\globalproj{Q}_1\ket{S}}$. 
By \Cref{cor:composition}
we deduce that the support of $\globalproj{Q}_{j+1}...\globalproj{Q}_1\ket{S}$
is bigger by the same factor, than that of $\globalproj{Q}_{j}...\globalproj{Q}_1\ket{S}.$
\end{proof}

We now observe that this expansion is too large 
to be applied for many layers, since the number 
of strings reached will just exceed the number of 
$n$-bit strings. 

\begin{lemma}[Bad string in constant number of layers]
\label{lem:constant-layers}
Let $\ell^* = \lceil\frac{2kd}{\eps} \log_{\left(1 +
  \frac{\eps}{4}\right)}|\Sigma|\rceil$. Consider a uniform\footnote{In fact, uniformity is not needed for this claim, but this requires more work.}
  stoquastic $k$-Local
  Hamiltonian $H = \frac{1}{m}\sum_{i = 1}^m H_i$ which is $\eps$-frustrated. Then for every good string $x$, there exists some $\ell < \ell^*$ and a sequence $L_1,...L_{\ell}$, where each $L_i$ consists of a set of non-overlapping local projectors (corresponding to the projections on the local terms of $H$), such that $L_{\ell}...L_1\ket{x}$ contains a bad string.
\end{lemma}
\begin{proof}
Let $L_1$ be the output of \Cref{algo:find-frustrated-terms} for the initial state $\ket{x}$, and recursively, for $\ell\le \ell^*$, let  
$L_{\ell}$ be the output of \Cref{algo:find-frustrated-terms} for the state $\reallywidehat{L_{\ell}...L_1\ket{x}}$. 
Let $S_{\ell}$ be the set of strings in the support of the 
state $\reallywidehat{L_{\ell}...L_1\ket{x}}$.

Now, if for some $\ell<\ell^*$, $S_{\ell}$ contains a bad 
string for $H$, we are done. Otherwise, 
for all $\ell<\ell^*$, the state 
$\reallywidehat{L_{\ell-1}...L_1\ket{x}}$ contains no bad 
string for $H$.

From \Cref{lem:expansion}, we have that for each $1 \leq \ell \leq \ell^*-1$, 
\begin{align*}\label{eq:bound-si}
|supp(L_{\ell}(\reallywidehat{L_{\ell-1}...L_1\ket{x}}))| \geq \left(1 + \frac{\eps}{4}\right)^{\frac{\eps n}{2kd}}|S_{\ell-1}| 
\end{align*}

Moreover, by \Cref{cor:composition}, 
we have 
\[|S_{\ell}|= |supp(L_{\ell}(\reallywidehat{L_{\ell-1}...L_1\ket{x}}))|.\]

Thus, by a trivial induction, we arrive at 

\begin{align*}
|S_{\ell^*}|\ge
\left( 
    \left(1 + \frac{\eps}{4}\right)^
    \frac{\eps n}{2kd}
                    \right)^
{\frac{2kd}{\eps}\log_{\left(1 + \frac{\eps}{4}\right)}|\Sigma|}
= 
\left( \left(1 + \frac{\eps}{4}\right)^{\log_{\left(1 + \frac{\eps}{4}\right)}|\Sigma|}\right)^n
=
|\Sigma|^n.
\end{align*}

Thus $S_{\ell^*}$ contains all possible strings. However if $H$ is frustrated, bad strings exist and thus there must be a bad string in 
$S_{\ell^*}= supp(L_{\ell^*}...L_1\ket{x})$.  
\end{proof}

\subsection{Finding the bad string}\label{sec:lightcone}
In this section, we prove that if for some constant $\ell$ we have that $L_{\ell}...L_1\ket{x}$ contains a bad string, and each $L_i$ consists of non-overlapping terms, then there exists a constant path (namely a sequence of constantly many local steps) from $x$ to a bad string. 

We start by showing how to retrieve (possibly polynomial-size) paths from strings in some non-negative state $\ket{\psi}$ to string in some state $L\ket{\psi}$, for a non-overlapping set of projections $L$.

\begin{lemma}[From non-overlapping projections to paths]
\label{lem:projections-to-path}
Let $\ket{\psi}$ be a non-negative state and $L$ be an arbitrary set of non-overlapping stoquastic
  projectors. Then for every string $y$ in $L\ket{\psi}$, there exists a string
  $x$ in $\ket{\psi}$ such that there is a $|L|$-step path between $x$ and $y$ in
  $G(H)$.
\end{lemma}
\begin{proof}
  Let  $L = \{\localproj{Q}_1, ...,
  \localproj{Q}_g\}$.
Since $L$ is a set of non-overlapping stoquastic projectors, if $y \in supp(L\ket{\psi})$, then there must exist some $x$ in $supp(\ket{\psi})$  such that
\[0 < \bra{x}L\ket{y}.\]
This is because writing $\bra{\psi}=\sum_x \psi_x \bra{x}$ we have (using $L^\dagger=L$ due to the fact that $L$ consists of non-overlapping projections): 
\[0 < \bra{\psi}L\ket{y}=\sum_{x\in supp(\ket{\psi})} \psi_x \bra{x}L\ket{y}\]
and the coefficients $\psi_x$ are all non-negative. 
We have
\begin{align}\label{eq:positive-summation}
0 < \bra{x}L\ket{y} =  
\bra{x}\globalproj{Q}_{1} ... \globalproj{Q}_{g}\ket{y} = \sum_{w_1, ..., w_{g-1}} \bra{x}\globalproj{Q}_{1}  \kb{w_1} \globalproj{Q}_2 ... \kb{w_{g-1}} \globalproj{Q}_{g} \ket{y},
\end{align}
where we can write the $\globalproj{Q}_i$ in any order since they are non-overlapping, from \Cref{rem:commuting}; the $w$'s above run over 
all $n$-dit strings. 

Since every $\globalproj{Q}_i$ has only non-negative entries, for every pair of strings $z$ and $z'$, $\bra{z}\globalproj{Q}_i\ket{z'} \geq 0$. Then \Cref{eq:positive-summation} holds iff there exists values $w^*_1, ..., w^*_{g-1}$  such that
\begin{align*}
0 < \bra{w^*_0}\globalproj{Q}_{1} \kb{w^*_1} \globalproj{Q}_{2} ... \kb{w^*_{g-1}} \globalproj{Q}_{g}\ket{w^*_g}, 
\end{align*}
where we have set $x := w^*_0$ and $y := w^*_g$. 
It follows that for every $i \in  [g]$, $\bra{w^*_{i-1}}\globalproj{Q}_{i}
  \ket{w^*_{i+1}} > 0$, and then from \Cref{eq:neighbor},  $w^*_i$ and
  $w^*_{i+1}$ are neighbors in $G(H)$. Therefore, there is a $|L|$-step path from $x$ to $y$.
\end{proof}

Finally, we show how to find a short path between the initial string a bad string. The intuition of the proof is depicted in \Cref{fig:lightcone}. 

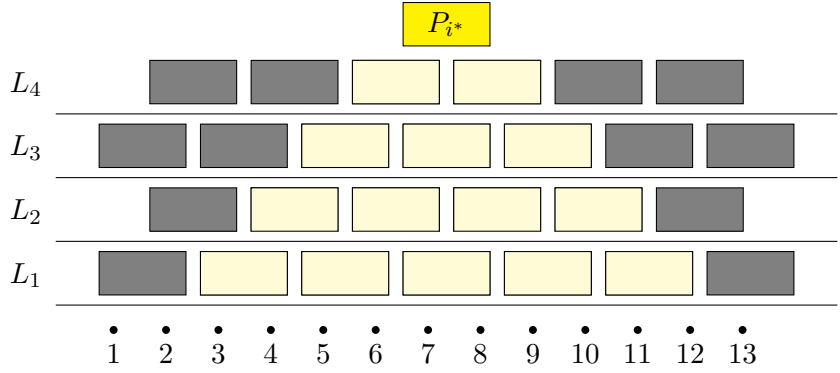
\begin{figure}
    \centering
\begin{tikzpicture}
    \node[circle,fill=black,inner sep=0pt,minimum
    size=3pt,label=below:{$1$}] (x1)  {};
  \foreach \i/\j in {2/1,3/2,4/3,5/4,6/5,7/6,8/7, 9/8, 10/9, 11/10,12/11,13/12}
  {
  \node[right=1.5em of x\j,circle,fill=black,inner sep=0pt,minimum
    size=3pt,label=below:{$\i$}] (x\i) {};   
  }
    \node[above=-1.5em of x1] (x1l0)  {};    
  \foreach \i/\j/\k in {0/1/2,2/3/4}
  {
    \node[above=1.5em of x1l\i] (x1l\j)  {};
    \node[above left=2.5em of x1l\i]  {$L_\j$};    
    \draw (x1l\j) ++(-2em,-0.35em) -- ++(27em,0);      
    \draw[draw=black,fill=gray] (x1l\j) ++ (-0.5em,0) rectangle ++(3em,1.5em);
    \draw[draw=black,fill=gray] (x1l\j) ++ (3em,0) rectangle ++(3em,1.5em);
    \draw[draw=black,fill=gray] (x1l\j) ++ (6.5em,0) rectangle ++(3em,1.5em);    
    \draw[draw=black,fill=gray] (x1l\j) ++ (10em,0) rectangle ++(3em,1.5em);    
    \draw[draw=black,fill=gray] (x1l\j) ++ (13.5em,0) rectangle ++(3em,1.5em);     
    \draw[draw=black,fill=gray] (x1l\j) ++ (17em,0) rectangle ++(3em,1.5em);         
    \draw[draw=black,fill=gray] (x1l\j) ++ (20.5em,0) rectangle ++(3em,1.5em);             
    \node[above left=2.5em of x1l\j]  {$L_\k$};        
    \node[above=1.5em of x1l\j] (x1l\k)  {};  
    \draw (x1l\k) ++(-2em,-0.35em) -- ++(27em,0);      
    \draw[draw=black,fill=gray] (x1l\k) ++ (1.25em,0) rectangle ++(3em,1.5em);
    \draw[draw=black,fill=gray] (x1l\k) ++ (4.75em,0) rectangle ++(3em,1.5em);
    \draw[draw=black,fill=gray] (x1l\k) ++ (8.25em,0) rectangle ++(3em,1.5em);    
    \draw[draw=black,fill=gray] (x1l\k) ++ (11.75em,0) rectangle ++(3em,1.5em); 
    \draw[draw=black,fill=gray] (x1l\k) ++ (15.25em,0) rectangle ++(3em,1.5em);  
    \draw[draw=black,fill=gray] (x1l\k) ++ (18.75em,0) rectangle ++(3em,1.5em);     
  }
      \draw[draw=black,fill=yellow!20] (x1l4) ++ (8.25em,0) rectangle ++(3em,1.5em);
      \draw[draw=black,fill=yellow!20] (x1l4) ++ (11.75em,0) rectangle ++(3em,1.5em);  
      \draw[draw=black,fill=yellow!20] (x1l3) ++ (6.5em,0) rectangle ++(3em,1.5em);      
      \draw[draw=black,fill=yellow!20] (x1l3) ++ (10em,0) rectangle ++(3em,1.5em);   
      \draw[draw=black,fill=yellow!20] (x1l3) ++ (13.5em,0) rectangle ++(3em,1.5em);            
      \draw[draw=black,fill=yellow!20] (x1l2) ++ (4.75em,0) rectangle ++(3em,1.5em);      
      \draw[draw=black,fill=yellow!20] (x1l2) ++ (8.25em,0) rectangle ++(3em,1.5em);      
      \draw[draw=black,fill=yellow!20] (x1l2) ++ (11.75em,0) rectangle ++(3em,1.5em);      
      \draw[draw=black,fill=yellow!20] (x1l2) ++ (15.25em,0) rectangle ++(3em,1.5em);      
      \draw[draw=black,fill=yellow!20] (x1l1) ++ (3em,0) rectangle ++(3em,1.5em);      
      \draw[draw=black,fill=yellow!20] (x1l1) ++ (6.5em,0) rectangle ++(3em,1.5em);      
      \draw[draw=black,fill=yellow!20] (x1l1) ++ (10em,0) rectangle ++(3em,1.5em);   
      \draw[draw=black,fill=yellow!20] (x1l1) ++ (13.5em,0) rectangle ++(3em,1.5em);          
      \draw[draw=black,fill=yellow!20] (x1l1) ++ (17em,0) rectangle ++(3em,1.5em);       

\draw[draw=black,fill=yellow] (x1l4) ++ (10em,2em) rectangle ++(3em,1.5em) node[pos=.5] {$\localproj{P}_{i^*}$};   
\end{tikzpicture}
\caption{The layers $L_1,..., L_4$ reach a bad string for some term $\localproj{P}_{i^*}$ that does not necessarily belong to any layer, and then we consider the light-cone of $\localproj{P}_{i*}$ in these layers. Notice that the light-cone is defined only within the layers, each consisting of non-overlapping terms, and {\it not} by considering {\it all} projectors which touch any qudit in $\localproj{P}_{i^*}$, then all projectors which touch those, etc.)}
\label{fig:lightcone-layers}
\end{figure}
 
\begin{lemma}[The light-cone argument]
\label{lem:lightcone}
For some initial string $x$, let $L_1, ..., L_{\ell}$ each be a sequence of non-overlapping projectors such that $L_{\ell}...L_1\ket{x}$ contains in its support a 
  string $w^*$ which is bad for $H$. Then, there is a $O(k^{\ell})$-steps path from $x$ to a bad string for $H$ (which could be different than 
  $w^*$). 
\end{lemma}
\begin{proof}
Since $w^*$ is bad for $H$, we have that 
for some $i^* \in [m]$, $\bra{w^*}\globalproj{P}_{i^*}\ket{w^*} = 0$.
Let $\lightcone{L}_{\ell} \subseteq L_{\ell}$ be the projectors of $L_{\ell}$
  that touch some qudit of $\localproj{P}_{i^*}$. 
We define recursively $\lightcone{L}_{j-1}$ as the set of projectors in
  $L_{j-1}$ that overlap with some projector in $\bigcup_{j' \geq j}
  \lightcone{L}_{j'}$. These are the {\it layers} of what we call the {\it
  light-cone} of $\localproj{P}_{i*}$ and we depict it in \Cref{fig:lightcone-layers}.
Let us also define the complement of $\lightcone{L}_j$ 
in $L_j$ to be $\outsidelightcone{L}_j$.  

For convenience, 
set $L_{\ell+1} = \lightcone{L}_{\ell+1} = \{\localproj{P}_{i^*}\}$. Let $D_j$ be the set 
of qudits touched by the terms in  $\bigcup_{j' \geq j} \lightcone{L}_{j'}$. We prove 
that  $|D_j|\le k^{\ell -j +2}$.
We prove this by a downward induction from $j = \ell+1$ to $j = 1$. The basis step
  is true since $\localproj{P}_{i^*}$ is $k$-local, and so $|D_{\ell+1}|=k$. 
  Now, assume this is true for 
  the $j$'th layer. $D_{j-1}$ is defined by adding to $D_j$ the qudits touched by the next layer of the lightcone,  $\lightcone{L}_{j-1}$.  
  By definition of the lightcone, these are 
  all qudits touched by terms in $L_{j-1}$
  that overlap $D_j$. 
  Since these terms are non-overlapping and 
  $k$-local, this 
  can at most multiply the number of qudits 
 already in $D_j$ by a factor of 
  $k$. 
We have that the set $D_{1}$ of qudits within the entire lightcone originating from 
$P_{i^*}$ contains at most 
$k^{\ell+1}$ qudits.

The terms in $L_j$ commute (as they are non-overlapping), and thus $L_{j}=\outsidelightcone{L}_{j}\lightcone{L}_{j}$. It follows that 
\begin{equation}\label{eq:layers}L_{\ell}....L_1\ket{x}= \outsidelightcone{L}_{\ell}\lightcone{L}_{\ell}\outsidelightcone{L}_{\ell-1}\lightcone{L}_{\ell-1}...\outsidelightcone{L}_1\lightcone{L}_{1}\ket{x} =
\outsidelightcone{L}_{\ell}\outsidelightcone{L}_{\ell-1}...\outsidelightcone{L}_1\lightcone{L}_{\ell}...\lightcone{L}_{1}\ket{x},\end{equation}
where in the second equality, we use in fact 
an iterative argument (that is common in light-cone reasoning, see, e.g.,  \cite{AharonovALV09,BravyiGK}): 
we notice that the projectors in $\outsidelightcone{L}_j$ commute with the projectors in $\lightcone{L}_{j'}$  
for all $j'\ge j$, and thus they can be commuted one by one across the lightcone operators, to the left. 
The fact that they commute follows by definition: 
level $j$ of the lightcone, $\lightcone{L}_{j}$, contains {\it all} terms in $L_{j}$ that overlap with any term $\lightcone{L}_{j'}$, for $j' > j'$; 
thus the remaining terms in $L_j$, namely 
$L'_j$, do not overlap the upper 
layers of the lightcone, and thus commute with them. 

From Equation \ref{eq:layers} we deduce that 
we can first apply on $x$ all terms 
in the lightcone, and delay all terms outside of the 
lightcone to later. From this, we can show that $\lightcone{L}_{\ell}...\lightcone{L}_{1}\ket{x}$ also contains a bad string for $\globalproj{P}_{i^*}$, and this will complete the 
proof. To do this, 
let $Q$ be the set of positions where the term  %
$P_{i^*}$ 
acts non-trivially. We claim that the application 
of the terms outside of the lightcone, which 
do not touch $Q$, couldn't have added a string which is bad with respect to $P_{i^*}$, unless 
such a string was there before. 
This can be deduced from \Cref{rem:locality-phi}, 
which when applied iteratively gives that 
\begin{align}\label{eq:containment-bad-string}
\{y_{Q} : y \in supp(\outsidelightcone{L}_{\ell}\outsidelightcone{L}_{\ell-1}...\outsidelightcone{L}_1\lightcone{L}_{\ell}...\lightcone{L}_{1}\ket{x})\} \subseteq \{y_{Q} : y \in supp(\lightcone{L}_{\ell}...\lightcone{L}_{1}\ket{x})\}. 
\end{align}

Since
  $\outsidelightcone{L}_{\ell}\outsidelightcone{L}_{\ell-1}...\outsidelightcone{L}_1\lightcone{L}_{\ell}...\lightcone{L}_{1}\ket{x}$
  contains a bad string $w^*$ for $\globalproj{P}_{i*}$, from
  \Cref{eq:containment-bad-string} we have that $\lightcone{L}_{\ell}...\lightcone{L}_{1}\ket{x}$ contains a string $w'$ such that $w^*_{Q} = w'_{Q}$, thus $w'$ is also bad for $\globalproj{P}_{i^*}$.

Finally, we can use \Cref{lem:projections-to-path} together with a (highly wasteful) bound on 
$|\lightcone{L}_j|\le |D_j|\le k^{\ell-j+2}$ recursively for
  $\lightcone{L}_{\ell}...\lightcone{L}_1\ket{x}$: for any string $y$ in
  $\lightcone{L}_j...\lightcone{L}_1\ket{x}$, there exists a string $y'$ in
  $\lightcone{L}_{j-1}...\lightcone{L}_1\ket{x}$, such that there is a $|\lightcone{L}_j|\le k^{\ell -j
  +2}$-step path from $y'$ to $y$ in $G(H)$. Hence,  there is a path of size
  $\sum_{j'=1}^{\ell} k^{\ell-j'+2} = \sum_{j'=2}^{\ell+1} k^{j'}=O(k^{\ell})$ from $x$ to $w'$ in $G(H)$.
\end{proof}

\subsection{Proof of Lemma \ref{lem:constant-step-path}}\label{sec:constant-path}
We can finally prove \Cref{lem:constant-step-path} by composing \Cref{lem:constant-layers,lem:lightcone}.

\newtheorem*{thm:repeat-constant}{\Cref{lem:constant-step-path} (restated)}
\begin{thm:repeat-constant}
\bodyconstantsteppath
\end{thm:repeat-constant}
\begin{proof}
Let $\ell^* = \frac{2kd}{\eps} \log_{\left(1 + \frac{\eps}{4}\right)}|\Sigma|$.
By \Cref{lem:constant-layers}, there exists a sequence of sets $L_1,...L_{\ell}$, for
  some $\ell \leq \ell^*$, such that each $L_i$ consists of non-overlapping projectors,
  and $ L_{\ell}...L_1\ket{x}$ contains a bad string. We finish the proof by using $L_1, ..., L_{\ell}$ in \Cref{lem:lightcone} for some $\ell \leq \frac{2kd}{\eps}
  \log_{\left(1 + \frac{\eps}{4}\right)}|\Sigma|$.
\end{proof}

\section{Negligible vs. high frustration}
 \label{sec:negligible}
In this section we show that we can replace the frustration-free property by allowing the stoquastic Hamiltonian to be negligibly-frustrated in yes-instances. 

Our proof follows by showing that in this case, there exists some
string $x$ in the groundstate of the Hamiltonian that cannot reach a bad string
in a constant number of steps, and therefore the same verification algorithm of \Cref{thm:main} would still work.

In order to prove that, we first show that if the groundenergy of the Hamiltonian is very small, then there is a also a subset state $\ket{S}$ with very low energy in respect to $H$. Then, in order
to achieve our goal, we need bound two quantities. First, we need to show that 
the number of bad strings in $\ket{S}$ is small. Secondly, we consider
the good strings that are connected to the strings in $\ket{S}$ but are not in
this state. We notice that these strings might be harmful, since they could be connected to bad strings and then strings in the ``border'' of $S$ could be close to bad strings. Finally, we can show that there must be some string in $S$ that is far from all of these strings if the frustration of the Hamiltonian is sufficiently small. This follows directly from the upper bound on the number of such strings and from the structure of $G(H)$.

\medskip 
Let us start then by bounding the weight of bad strings in the groundstate of slightly
frustrated Hamiltonians.

\begin{lemma}\label{lem:connection-frustration-bad-strings}
Let $H$ be a stoquastic Hamiltonian  and $\ket{\psi} = \sum_{x} \alpha_x \ket{x}$ be a positive state. Let $B$ be the set of bad strings for $H$. Then $\sum_{x \in B} \alpha_x^2 \leq m\bra{\psi}H\ket{\psi}$
\end{lemma}
\begin{proof}
Let $B_i$ be the set of bad strings for some term $H_i$, 
$\ket{\psi} = \alpha \ket{\psi_{G_i}} + \beta \ket{\psi_{B_i}}$ where
$\alpha = \sqrt{\sum_{x \not\in B_i} \alpha_x^2}$, $\beta = \sqrt{\sum_{x \in B_i} \alpha_x^2}$, 
$\ket{\psi_{G_i}} = \frac{1}{\alpha} \sum_{x \not\in B_i} \alpha_x \ket{x}$ and $\ket{\psi_{B_i}} = \frac{1}{\beta} \sum_{x \in B_i} \alpha_x \ket{x}$. 

We have that
\begin{align*}
    \bra{\psi}H_i\ket{\psi} &
    = (\alpha\bra{\psi_{G_i}}+ \beta\bra{\psi_{B_i}})(I - P_i)(\alpha\ket{\psi_{G_i}}+ \beta\ket{\psi_{B_i}}) \\
    &=
    1 - \alpha^2 \bra{\psi_{G_i}}P_i\ket{\psi_{G_i}} \\
    &\geq
    1 - \alpha^2 \braket{\psi_{G_i}}{\psi_{G_i}} \\
    &=  \beta^2.
\end{align*}
where in the first equality we assume that $H_i$ is a projector (see \Cref{rem:term-projector}) and in the second inequality we use the fact that that $\ket{\psi_{G_i}}$ is positive and $P_i$ is stoquastic.

We have then that
\begin{align*}
\bra{\psi}H\ket{\psi} \geq \frac{1}{m} \sum_i \sum_{x \in B_i} \alpha_x^2 
\geq \frac{1}{m} \sum_{x \in B} \alpha_x^2
\end{align*}
where we use in the second inequality 
that each $x\in B$ appears in at least one of the $B_i$. 
\end{proof}

\begin{lemma}[Projection on boundary lower bounds energy]\label{lem:connection-frustration-missing-strings}
Let $H$ be a stoquastic Hamiltonian  and $\ket{\psi} = \sum_{x} \alpha_x \ket{x}$ be a positive state. Let $N$ be the set of $\{x \in  supp(\ket{\psi} : \exists i \in [m], y \not\in supp(\ket{\psi}), \bra{x}P_i\ket{y} > 0\}$. Then $\bra{\psi}H\ket{\psi} \geq \frac{1}{|\Sigma|^k m} \sum_{x \in N} \alpha_x^2$.
\end{lemma}
\begin{proof}
Let $N_i = \{x \in  supp(\ket{\psi} : \exists y \not\in supp(\ket{\psi}), \bra{x}P_i\ket{y} > 0\}$, 
$\ket{\psi} = \alpha \ket{\psi_{M_i}} + \beta \ket{\psi_{N_i}}$ where
$\alpha = \sqrt{\sum_{x \not\in N_i} \alpha_x^2}$, $\beta = \sqrt{\sum_{x \in N_i} \alpha_x^2}$, 
$\ket{\psi_{M_i}} = \frac{1}{\alpha} \sum_{x \not\in N_i} \alpha_x \ket{x}$ and $\ket{\psi_{N_i}} = \frac{1}{\beta} \sum_{x \in N_i} \alpha_x \ket{x}$.

We have that
\begin{align}
    \bra{\psi}H_i\ket{\psi} &
    =(\alpha\bra{\psi_{M_i}}+ \beta\bra{\psi_{N_i}})(I - P_i)(\alpha\ket{\psi_{M_i}}+ \beta\ket{\psi_{N_i}}) \nonumber \\
    &= \label{eq:middle-missing-strings}
    1 - \alpha^2 \bra{\psi_{M_i}}P_i\ket{\psi_{M_i}} -  \beta^2\bra{\psi_{N_i}}P_i\ket{\psi_{N_i}}
\end{align}
where in the first equality we use the fact that if
$\bra{x}P_i\ket{y} > 0$ and $x\in N_i$, then 
there exists a $z \not\in supp(\ket{\psi})$
with
$\bra{x}P_i\ket{z} > 0$, which implies that 
$\bra{y}P_i\ket{z} > 0$ and therefore $y \in N_i$.

We will show below that
\begin{equation}\label{eq:bound}
    \bra{\psi_{N_i}}P_i\ket{\psi_{N_i}} \leq \frac{|\Sigma|^k-1}{|\Sigma|^k}.\end{equation}
Using this and     \Cref{eq:middle-missing-strings} we have
\begin{align*}
    \bra{\psi}H_i\ket{\psi} = 1 - \alpha^2 \bra{\psi_{M_i}}P_i\ket{\psi_{M_i}} -  \beta^2\bra{\psi_{N_i}}P_i\ket{\psi_{N_i}}
    \geq 1- \alpha^2 
    - \frac{|\Sigma|^k-1}{|\Sigma|^k}\beta^2 =   \frac{1}{|\Sigma|^k} \beta^2 = \frac{1}{|\Sigma|^k}\sum_{x \in N_i} \alpha_x^2.
\end{align*}

This implies the desired claim since 
\begin{align*}
\bra{\psi}H\ket{\psi} \geq \frac{1}{m |\Sigma|^k} \sum_i \sum_{x \in N_i} \alpha_x^2 
\geq \frac{1}{m|\Sigma|^k}\sum_{x \in N} \alpha_x^2 
\end{align*}
where we used the fact that $N\subseteq \cup_i N_i$.

It remains to prove \Cref{eq:bound}. 
Consider some $T \subseteq \{0,1\}^n$ and positive $\ket{\phi} = \sum_{x \in T} \beta_x \ket{x}$, it follows that
\begin{align}\label{lem:projection-missing-string}
\braket{\phi}{T}\braket{T}{\phi} &= \sum_{x,y \in T} \frac{\beta_x\beta_y}{|T|}\leq 
\sqrt{\frac{|T \cap supp(\ket{\phi})|^2}{|T|^2} \sum_{x,y \in  T} \beta_x^2 \beta_y^2} = \frac{|supp(\ket{\phi})|}{|T|},
\end{align}
where we use Cauchy-Schwartz inequality.

Given the decomposition $P_i = \sum_{j,z} \kb{T_{i,j,z}}$, we can define $\ket{\psi_{N_i}} = \gamma_{j,z} \sum_{j,z}\ket{\psi_{N_i,j,z}}$, 
with $supp(\ket{\psi_{N_i,j,z}}) \subseteq T_{i,j,z}$ and then 
\begin{align}\label{eq:conclusion-missing-strings}
\bra{\psi_{N_i}}P_i\ket{\psi_{N_i}} = \sum_{j,z} \gamma_{j,z}^2 \braket{\psi_{N_i,j,z}}{T_{i,j,z}}
\braket{T_{i,j,z}}{\psi_{N_i,j,z}} \leq \sum_{j,z}
\gamma_{j,z}^2 \frac{|supp(\ket{\psi_{N_i,j,z}})|}{|T_{i,j,z}|} \leq \frac{|\Sigma|^k-1}{|\Sigma|^k},
\end{align}
where the first inequality follows from
\Cref{lem:projection-missing-string} and in the second inequality we use that $|T_{i,j,z}| \leq |\Sigma|^k$ and $|supp(\ket{\psi_{N_i,j,z}})| \leq |T_{i,j,z}| - 1$ (by definition of $N_i$).
\end{proof}

We now prove that there exists a low-energy state with enough structure.

\begin{lemma}[Existence of a "nice" low energy state] \label{lemma:structure-gstate-negligible}
Let $H$ be a stoquastic Hamiltonian with non-zero ground-energy at most $\frac{1}{f(n)}$. Then there exists some non-negative state $\ket{\psi}$ that does not contain bad strings for $H$, contains only strings with amplitude at least $\delta = \frac{1}{\sqrt{g(n)|S|}}$ and has frustration at most $\frac{1}{f(n)(1 - \frac{m}{f(n)} - \frac{1}{g(n)} )}$ .
\end{lemma}
\begin{proof}
Let $\ket{\phi} = \sum_{x} \alpha_x \ket{x}$ be a non-negative groundstate of $H$, $S = supp(\ket{\phi})$,  $B$ be the set of bad strings, $L$ be the set of strings with amplitude smaller than $\delta$, $T = B \cup L$ and $G = S \setminus T$. We define $\ket{\psi} = \frac{1}{\sqrt{\sum_{x\in G} \alpha_x^2}}\sum_{x \in G} \alpha_x \ket{x}$.
We have that: 

\begin{align}
    \bra{\phi}H\ket{\phi} &= \sum_{x,y\in G}\alpha_x \alpha_y\bra{x}H\ket{y} + 
    \sum_{x\in T, y\in S}\alpha_x \alpha_y\bra{x}H\ket{y}
    + \sum_{x\in G, y\in T}\alpha_x \alpha_y\bra{x}H\ket{y}  \\
    &\geq \sum_{x,y\in G}\alpha_x \alpha_y\bra{x}H\ket{y}  \\
    &= 
\frac{\sum_{x\in G} \alpha_x^2}{\sum_{x\in G} \alpha_x^2} \sum_{x,y\in G}\alpha_x \alpha_y\bra{x}H\ket{y} \\
& = \left(1-\sum_{x \in B} \alpha_x^2 - \sum_{x \in T\setminus B} \alpha_x^2\right) \bra{\psi}H\ket{\psi} \\
&\geq \left(1-\frac{m}{f(n)} - |S|\delta^2 \right) \bra{\psi}H\ket{\psi},
\end{align}
where the last inequality comes from 
\Cref{lem:connection-frustration-bad-strings} and the fact that $\sum_{x \in T\setminus B} \alpha_x^2 \leq \sum_{x \in T\setminus B} \delta^2 \leq |S|\delta^2$.  
\end{proof}

We finally show that the structured low-energy state implies that there is a string that is ``far from the frustration''.

\begin{lemma}[Nice low energy state has small boundary]\label{lemma:sstate-negligible}
Let $H$ be a stoquastic Hamiltonian such that there exists some non-negative state $\ket{\psi}$  such that $S := supp(\ket{\psi})$ does not contain bad strings for $H$, and all amplitudes of $\ket{\psi}$ are at least $\delta = \frac{1}{\sqrt{g(n)|S|}}$. Moreover, it has frustration strictly smaller than $\frac{1}{h(n)}$. Then it follows that $|\{ x \in S: \exists y \not\in S, i \in [m] \text{ s.t. } \bra{x}P_i\ket{y} > 0 \}| < \frac{m|\Sigma|^k g(n)}{h(n)} |S|$.
\end{lemma}
\begin{proof}
Let $N = \{ x \in S: \exists y \not\in S, i \in [m] \text{ s.t. } \bra{x}P_i\ket{y} > 0 \}$. 
Let us assume towards contradiction that $\frac{|N|}{|S|} > \frac{m |\Sigma|^k g(n)}{h(n)}$. We have 
from \Cref{lem:connection-frustration-missing-strings}
\begin{align*}
   \frac{1}{h(n)} > \bra{\psi}H\ket{\psi} \geq \frac{1}{m|\Sigma|^k} \sum_{x \in N} \alpha_x^2 \geq \frac{|N| \delta^2}{m|\Sigma|^k} > 
   \frac{m\delta^2 g(n) |S|}{m h(n)}=
   \frac{1}{h(n)},
\end{align*}
where the next to last inequality follows since 
the set of strings in $N$ are all in $S$ the support of $\ket{\psi}$ and thus their amplitudes by assumption are at least $\delta$.
This is a contradiction.

\end{proof}

\begin{lemma}[A state of small boundary has a ``protected" string]
  \label{lem:new-completeness}
  Fix some parameter $t \in \mathbb{N}^*$.
Let $H$ be a uniform stoquastic $k$-local Hamiltonian and $\ket{\psi}$ be a state such that its support $S$ contains no bad strings for $H$, and $|\{ x \in S: \exists y \not\in S, i \in [m] \text{ s.t. } \bra{x}P_i\ket{y} > 0 \}| < \frac{1}{(|\Sigma|^k m)^{t+1}}|S|$. Then there exists some string $x \in S$ such that all
  strings within distance $t$ from $x$ in the graph $G(H)$ are $H$-good strings.
\end{lemma}
\begin{proof}
Let $M = \{ y \not\in S: \exists x \in S, i \in [m] \text{ s.t. } \bra{x}P_i\ket{y} > 0 \}$ be the set of strings not in $S$ that are connected with strings in $S$.
Notice that $|M| \leq |\Sigma|^k m|\{ x \in S: \exists y \not\in S, i \in [m] \text{ s.t. } \bra{x}P_i\ket{y} > 0 \}| < \frac{1}{(|\Sigma|^k m)^{t}}|S|$. 

Also, for any string $x$, we have that there exists at most $(|\Sigma|^k m)^{t}$ strings within distance $t$ steps from $x$.  
Therefore the number of strings that could reach some string in $M$ in $t$ steps is strictly smaller than 
  $\frac{|S|(|\Sigma|^k m)^{t}}{(|\Sigma|^k m)^{t}} = |S|$. 
  
  This means that there must exist a string in $S$ from which a $t$ step walk cannot lead to a string outside $S$.
\end{proof}

Finally, we prove the main technical result of this section.

\begin{theorem}
  For any constant $\eps>0$, let $t = \sizepath$ (with constants defined as in Definition \Cref{def:local-hamiltonian}) and $c$ be any function such
  that $c(n) < \frac{1}{(|\Sigma|^k m)^{2(t+2)}}$,
  the problem of deciding whether a uniform stoquastic
  Hamiltonian $H$ has ground-energy at most $c(n)$ or at least
  $\eps$ is in \NP{}.
\end{theorem}
\begin{proof}
The \NP{} witness for the problem consists of some initial string $x$ that is
  promised to be in the support of some groundstate $\ket{S}$ of $H$ and such
  that all strings within distance $t$ of $x$ are good and lie in $S$ .
  The verification proceeds by running over all
  possible $t$-step paths from $x$. As argued in \Cref{thm:main} such
  verification can be performed in polynomial time.

\begin{align*}
    &|\{ x \in S: \exists y \not\in S, i \in [m] \text{ s.t. } \bra{x}P_i\ket{y} > 0 \}| \\
    &< \frac{m |\Sigma|^k \sqrt{f(n)}}{(1-\frac{m}{f(n)}-\frac{1}{g(n)})f(n)} |S| \\
    &\leq \frac{m |\Sigma|^k \sqrt{c(n)}}{1-2\sqrt{c(n)}m} |S| \\  
    &\leq \frac{(|\Sigma|^k m)^{(t+2)}}{(|\Sigma|^k m)^{t+2} \left((|\Sigma|^k m)^{t+2} -2m\right)}|S| \\    
    &\leq \frac{|S|}{(|\Sigma|^k m)^{t+1}}.
\end{align*}
Therefore, by \Cref{lem:new-completeness}, we have that for yes-instances, there must
  exist a string $x$ that does not reach bad strings in $t$ steps, and therefore the prover can send this string to verifier.

On the other hand, if $H$ is $\eps$-frustrated, then  by \Cref{lem:constant-step-path}, there must exist a $t$-step path from $x$ to some string $y$ that is bad for $H$, and such path will be found by the brute-force search.
\end{proof}

This proves~\Cref{thm:exponentially-small}.

\appendix

\section{Results related to derandomization of RP}
\label{sec:coRP}
We start with the class co-RP, which is the class of problems that can be solved by randomized algorithms with perfect completeness.

\defclass{co-RP}{def:coRP}{
A problem $A=(\ayes,\ano)$ is in co-RP if and only if there exists
 a probabilistic polynomial-time algorithm $R$, where $R$
takes as input a string $x\in\Sigma^*$ and decides
on acceptance or rejection of $x$
such that:  
}
{ If $x\in\ayes$, then $R$ accepts $x$ with probability $1$.}
{ If $x\in\ano$, then $R$
    accepts $x$ with probability at most $\frac{1}{2}$.}
    
The only difference of \Cref{def:coRP,def:MA} is that in MA there is some witness $y$, unknown by the verification algorithm. We can see it as if the witness in co-RP is trivially the empty string. In this case, it is not surprising that we can define a version of \Cref{def:local-hamiltonian}, where we fix some string in the groundstate.

\defproblem{pinned uniform stoquastic frustration-free  $k$-Local Hamiltonian problem}{def:pinned-local-hamiltonian}{
  The {\em pinned uniform stoquastic  frustration-free $k$-Local Hamiltonian} problem,
  where $k \in \mathbb{N}^*$ is called the locality and $\eps :  \mathbb{N} \rightarrow [0,1]$ is a non-decreasing function,
  is the following promise problem. Let $n$ be the number of qudits of a quantum system.
  The input is a set of $m(n)$ uniform stoquastic Hamiltonians $H_1, \ldots, H_{m(n)}$
  where $m$ is a polynomial, $\forall i \in m(n) : 0 \leq H_i \leq I$
  and each $H_i$ acts on $k$ qudits out of the $n$ qudit system. We also assume that there are at most $d$ terms that act non-trivially on each qudit, for some constant $d$.
  For $H = \frac{1}{m(n)} \sum_{j = 1}^{m(n)} H_j$ , one of the following two conditions holds, and the problem is to decide which one: 
}
{There exists a $n$-qudit quantum state
      state $\ket{\psi}$ such that $\braket{\psi}{0} > 0$ and
      $\bra{\psi} H \ket{\psi}
        =0$.}
{For all $n$-qudit quantum states $\ket{\psi}$,      it holds that
      $\bra{\psi} H \ket{\psi}
        \geq \eps(n) .$
        }

\begin{theorem}\label{thm:complete-corp}
The pinned uniform stoquastic frustration-free  $k$-Local Hamiltonian problem is
  in co-RP for every constant $k$ and $\eps(n) = \frac{1}{p(n)}$ where $p$ is
  some polynomial. Also, for some $p'(n) = O(n^2)$, the pinned uniform
  stoquastic frustration-free $6$-Local Hamiltonian problem is co-RP complete.
\end{theorem}
\begin{proof}
The inclusion in co-RP comes directly from the random-walk proposed by Bravyi and Terhal~\cite{BravyiT09}, starting from the fixed all-zeros string.

For the hardness part, as in~\cite{BravyiT09}, we can analyze the reduction of
  the quantum Cook-Levin theorem for a co-RP circuit. In the yes-instance, we
  have that the all-zeros string must be in the groundstate, since it is a valid initial configuration, whereas for no-instances, all states have inverse-polynomial frustration. 
\end{proof}

\begin{theorem}\label{thm:complete-corp-P}
The $\epsilon$-gapped pinned $k$-local $d$-bounded-degree uniform stoquastic frustration-free Hamiltonian problem is
  in P for every constants $k,d$ and $\eps>0$.
\end{theorem}

\begin{proof}
The proof follows exactly the lines of 
the proof of Theorem \ref{thm:main} except 
instead of starting the walk from the string 
given by the prover, the algorithm treats the 
all $0$ string as the witness, and hence 
no prover is needed. We omit the details. 
\end{proof}

\begin{conjecture}[Witness preserving Stoquastic PCP conjecture] (Informal) 
\label{conj:pcpWitnessPreserving}
There exist constants $\eps > 0$, $k',d' > 0$ and an efficient gap amplification procedure that reduces the problem deciding if a uniform stoquastic $d$-degree $k$ Local-Hamiltonian is frustration-free or at least inverse polynomially frustrated, to the problem of deciding if a uniform stoquastic $d'$-degree $k'$ Local-Hamiltonian is frustrarion-free or at least $\eps$ frustrated.
In the yes-cases, 
it also provides an efficient mapping from any
witness of the original problem to a witness 
of the target gapped problem. 
\end{conjecture}

\begin{corollary}
Conjecture \ref{conj:pcpWitnessPreserving}  implies $\RP{}=\cP{}$\label{cor:RP}.
\end{corollary}

\begin{proof} (sketch)
We use the stoquastic gap amplification 
procedure to generate an instance of the gapped version of the Hamiltonian, and consider $x$, the image of the all-$0$ string 
under the witness preserving reduction.
Then we apply the same algorithm which
the verifier of Theorem \ref{thm:main} does, but instead of applying it on the witness,
apply it on the string $y$. 
\end{proof}

We do not know how to extend these results to BPP without making highly artificial assumptions, since for the proof of Theorem \ref{thm:exponentially-small}
to go through, we need the pinned string to not only be in the support of the groundstate but also be far from bad strings. Hence making interesting implications for the derandomization of BPP remains for future work.

\section{Commuting Stoquastic Hamiltonians are in NP}
\label{sec:commuting}
In this section, we give a simple proof that deciding if a commuting Stoquastic Hamiltonian is frustration-free is in NP. This simplifies  the proof of this statement originally given  in~\cite{GharibianLSW15}.

\newtheorem*{thm:repeat-commuting}{\Cref{thm:commuting} (restated)}
\begin{thm:repeat-commuting}
\bodycommutingthm
\end{thm:repeat-commuting}

We notice that given that we are not assuming anything on the gap, one needs to be somewhat careful with the assumptions on how the input 
is given; we assume here that the local terms of the commuting stoquastic
  Hamiltonian $H$ are provided by giving the 
  matrix elements, each 
  with $poly(m)$ bits, and the terms 
  mutually commute exactly. 

\begin{proof}
Let $\localproj{H}_1,...,\localproj{H}_m$ be the terms in the Hamiltonian, and let $\globalproj{P}_1,...,\globalproj{P}_m$ be the
  corresponding projectors onto their groundspaces. We show that 
$H$ is frustration-free iff there exists a string $x$ that is good for $H$.

The first direction $(\Rightarrow)$ is trivial: any string in the $0$-energy groundstate is good for $H$.

We prove now the converse $(\Leftarrow)$. Let $x$ be a string that is good for $H$.
Let $\ket{\phi} = \globalproj{P}_{1}\globalproj{P}_{2}...\globalproj{P}_{m}\ket{x}$.
Since $\globalproj{P}_1, ..., \globalproj{P}_m$ commute, we have that 
\[\ket{\phi} = \globalproj{P}_i \ket{\phi}, \]
which means that either $\ket{\phi} = 0$ or it is a $+1$-eigenstate of every $\globalproj{P}_i$, and therefore it has energy $0$ with respect to $H$.
We show now that $\ket{\phi} \ne 0$. The string $x$ is in the support of some groundstate of every term $H_i$, therefore for every state $\ket{\alpha} = \sum_{y}\alpha_y \ket{y}$ with $\alpha_y \in \real^+$ and $\alpha_{x} > 0$, $\globalproj{P}_i\ket{\alpha} =  \sum_{y}\alpha'_y \ket{y}$, with $\alpha'_y \in \real^+$ and $\alpha'_{x} > 0$. Therefore we have that $\globalproj{P}_{1}\globalproj{P}_{2}...\globalproj{P}_{m}\ket{x}  \ne 0$.

Finally, to show that the problem is in NP: the proof is supposed to be a string with largest amplitude in some groundstate of the Hamiltonian. The verification algorithm checks if this string is indeed good for all local terms, as follows. First, for each $k$-local term $H_i$, the verifier computes $P_i$ to within constant approximation. More precisely, it computes each matrix elements of $P_i$ 
to within $\frac{1}{4|\Sigma|^{k}}$. 
This can be done efficiently since we are working with local terms where each 
matrix element of $H_i$ is specified by polynomially many bits. Let the approximated projector be $P'_i$. 
The verifier then checks if $x$ is bad for $P_i$: to do this, it first restricts $x$ to $Q$, the set of $k$ qudits on which $H_i$ acts, and then checks if  $\bra{x_Q}P'_i\ket{x_Q}\le \frac{1}{2|\Sigma|^k}$. If yes, then it rejects. If the verifier does not reject for any of the $i$'s, then it accepts. 

To finish the proof, we need to show that the above procedure cannot lead to an error in the verifier's decision of whether $x$ is bad or good for $H$. We start by arguing that in the frustration-free case, if $x$ is the string with largest amplitude in some groundstate $\ket{\psi}$, then it passes the test.  
This follows since we can write $\tilde{P}_i$ using its non-negative decomposition (see \Cref{rem:notation}): 
$\tilde{P}_i= \sum_{j} \kb{\globalproj{\phi}_{i,j,z}}$. We then write $\ket{\psi} = \sum_{j,z} \alpha_{i,j,z} \ket{\globalproj{\phi}_{i,j,z}}$; there is a unique  $\ket{\globalproj{\phi}_{i,j^*,z^*}}$ that contains $x$. Then $x$ must be the string with largest amplitude in $\ket{\globalproj{\phi}_{i,j^*,z^*}}$ 
(since a string $y$ in   $\ket{\globalproj{\phi}_{i,j^*,z^*}}$
has amplitude $\alpha_{i,j,z} \braket{y}{\globalproj{\phi}_{i,j,z}}$ in $\ket{\psi}$ and $x$ maximizes this value). Since $\ket{\globalproj{\phi}_{i,j^*,z^*}}$ contains at most $|\Sigma|^k$ strings, the amplitude of $x$ in it is at least $|\Sigma|^{-\frac{k}{2}}$. Hence
 $\bra{x}\globalproj{P}_i\ket{x} =  \braket{x}{\globalproj{\phi}_{i,j^*,z^*}}\braket{\globalproj{\phi}_{i,j^*,z^*}}{x} \geq \frac{1}{|\Sigma|^{k}}$.
 By our described procedure, in this case the verifier accepts. 
 For soundness, if some string $x$ is bad for $H_i$, then for some $i$ 
$\bra{x}P_i\ket{x} = 0$ and thus by our bound on the error due to approximation of $P_i$ by $P'_i$, we know  
$\bra{x}P'_i\ket{x} \leq \frac{1}{4|\Sigma|^{k}}$, and the verifier will reject. 
\end{proof}

\section{Classical definition of the problem}\label{app:classical}

As we discussed in \Cref{sec:discussion}, we can also describe the uniform stoquastic frustration-free local Hamiltonian problem in a more classical language. In this section, we present such approach. 

\medskip 

We start by defining what a constraint is in our setting.

\begin{definition}[$k$-SetConstraint]
A $k$-SetConstraint is a tuple $(T,B)$, where $B \subseteq [n]$ with $|B| = k$ and $T = \{T_1,...,T_l\}$ where $T_{j} \subseteq \Sigma^k$ and  $T_{j} \cap T_{j'} = \emptyset$ for $j \ne j'$. 
\end{definition}

We define now what it means for a constraint to be satisfied by a set of strings. In order to do it, we need to present some notation. For $x \in \Sigma^n$ $S \subseteq \Sigma^n$, $B \subset [n]$, and $y \in \Sigma^{n-|B|}$ we denote $x|_{B}$ as the substring of $x$ on the positions contained in $B$, $S_B = \{x|_B : x \in S\}$ and $S_{B,y} = \{x_B : x \in S \text{ and } x_{\overline{B}} = y\}$.

\begin{definition}[Satisfiability of set-constraints]
We define the UNSAT-value of a $k$-SetConstraint $(T,B)$ with respect to a subset of strings $S \subseteq \Sigma^n$, to be 
\begin{align}\label{unsat}
UNSAT((T,B),S) =
1-\sum_{T_j \in T}\sum_{y\in \Sigma^{n-k}}  \frac{|T_{j} \cap S_{B,y}|^2}{|T_{j}||S|}.
\end{align}
A sequence of $m$ $k$-SetConstraints $I = ((T_1,B_1),...,(T_{m},B_{m}))$ has $\eps$-UNSAT value with respect to $S \subseteq \Sigma^n$, if 
\begin{align}
UNSAT(I,S) \defeq \frac{1}{m}\sum_{i=1}^m UNSAT((T_i,B_i),S) \geq \eps.
\end{align}
We say that $I$ is satisfiable if there exists an $S\subseteq \Sigma^n$ such that $UNSAT(I,S)=0$.
We say that $I$ is $\epsilon$-frustrated if for all 
$S\subseteq \Sigma^n$ $UNSAT(I,S)\ge \epsilon$.
\end{definition}

It is not difficult to see that 
the UNSAT value by Equation \ref{unsat} is a value between 
$0$ and $1$. 
As in standard CSP, an instance to the problem consists of a sequence of constraints and we ask if they can all be satisfied simultaneously.

\defproblem{$k$-local Set Constraint Satisfaction problem (\problem{SetCSP})}{def:set-csp}{ Fix $k>0, d>0$. 
For a function $\epsilon(n)$, an instance to the {\em $\epsilon, d$ $k$-local Set Constraint Satisfaction problem} is a sequence of $m(n)$ $k$-SetConstraints
  $I = \{(T_1,B_1),...,(T_{m(n)},B_{m(n)})\}$ on $\Sigma^n$, 
  where $m$ is some polynomial in 
  $n$.
  We also assume that every $j \in [n]$ appears on at most $d$ $B_i$'s. Under the promise 
  that one of the following two holds, decide whether: 
}
{There exists some $S \subseteq \Sigma^n$ such that $UNSAT(I,S)=0$}
{For all $S \subseteq \Sigma^n$, 
$UNSAT(I,S)\ge epsilon$.}

We show now the equivalence 
between the uniform stoquastic Local Hamiltonians problem and the \problem{SetCSP} problem. 

\begin{lemma}
There is a polynomial time 
computable transformation which 
given an instance to the 
$(\epsilon, k,d)$ uniform gapped 
stoquastic local Hamiltonian problem, outputs an instance to the 

$\epsilon, k,d$-\problem{SetCSP}, such that the frustration is preserved. Similarly, there is also a frustration preserving reduction between the latter problem to 
the former. 
\end{lemma}
\begin{proof}
Let $H = \frac{1}{m} \sum_i H_i$ be an instance of the uniform stoquastic Hamiltonian problem. 

Let $P_i = \sum_{ij} \kb{T_{i,j}}$ be the projector onto the groundspace of $H_i$ for disjoint $T_{i,j} \subseteq \Sigma^k$. For each local term $H_i$, the \problem{SetCSP} instance has a constraint $(T_i, B_i)$ where $T_i = \{T_{i,j}\}$ and $B_i$ is the set of qubits on which $H_i$ acts non-trivially. 
Denote by $\sigma(H)$ the 
\problem{SetCSP} instance $I = ((T_1,B_1),...,(T_m,B_m))$. 

It can easily be checked that for $S \subseteq \Sigma^n$, we have that
\begin{align*}
\bra{S}H\ket{S} = 1 - \frac{1}{m}\bra{S}\globalproj{P}_i\ket{S} = 1 - 
\frac{1}{m}\sum_i \sum_{j,z} \braket{S}{\tilde{T}_{i,j,z}}\braket{\tilde{T}_{i,j,z}}{S} = UNSAT(\sigma(H),S),
\end{align*}
where we used the definition of $\tilde{P}_i$ and $\tilde{T}_{i,j,z}$ as in~\Cref{rem:notation}.

To move in the other direction, 
simply define $H_i=I-P_i$, 
which can easily be seen 
to be uniform stoquastic. 
\end{proof}

The previous lemma shows that \problem{SetCSP} and uniform stoquastic Hamiltonians problem are equivalent. Thus, we have
by Bravyi and Terhal \cite{BravyiT09}: 

\begin{corollary}
There is an inverse polynomial function $\eps$ such that the $\epsilon,k,d$ \problem{SetCSP} is \MA-complete.
\end{corollary}

And by \Cref{thm:main}, we have the following..
\begin{corollary}
For any constant $\eps > 0$, the $k$-local Set Constraint Satisfaction problem is in \NP.
\end{corollary}

\bibliographystyle{alpha}
\bibliography{stoquastic}

\end{document}